\documentclass[english,unicode=true,copyright,creativecommons,final,fleqn]{eptcs}
\usepackage[T1]{fontenc}
\usepackage[latin9]{inputenc}
\usepackage{array}
\usepackage{url}
\usepackage{amsthm}
\usepackage{amsmath}
\usepackage{amssymb}

\makeatletter

\providecommand{\tabularnewline}{\\}

\hypersetup{pdfusetitle,bookmarks=true,bookmarksnumbered=true,bookmarksopen=true,bookmarksopenlevel=2,breaklinks=false,pdfborder={0 0 0},pdfborderstyle={},backref=false,colorlinks=false}

\usepackage{etoolbox}

\newcommand\lyxeptcstitlerunning[2]{\title{#2}\ifstrempty{#1}{}{\def\titlerunning{#1}}}
\newcommand\lyxeptcsauthorrunning[2]{\author{#2}\ifstrempty{#1}{}{\def\authorrunning{#1}}}
\newcommand\lyxeptcsevent[1]{\providecommand{\event}{#1}}
\usepackage{enumitem}		
\newlength{\lyxlabelwidth}      
\theoremstyle{plain}
\newtheorem{thm}{\protect\theoremname}
\theoremstyle{definition}
\newtheorem{defn}[thm]{\protect\definitionname}
 \theoremstyle{remark}
 \newtheorem*{interpretation*}{\protect\interpretationname}
\theoremstyle{plain}
\newtheorem{lem}[thm]{\protect\lemmaname}
\theoremstyle{remark}
\newtheorem{rem}[thm]{\protect\remarkname}
\theoremstyle{plain}
\newtheorem{prop}[thm]{\protect\propositionname}
\theoremstyle{plain}
\newtheorem{cor}[thm]{\protect\corollaryname}
\ifx\proof\undefined
\newenvironment{proof}[1][\protect\proofname]{\par
\normalfont\topsep6\p@\@plus6\p@\relax
\trivlist
\itemindent\parindent
\item[\hskip\labelsep
\scshape
#1]\ignorespaces
}{%
\endtrivlist\@endpefalse
}
\providecommand{\proofname}{Proof}
\fi
\newcounter{casectr}
\newenvironment{caseenv}
{\begin{list}{{\itshape\ \protect\casename} \arabic{casectr}.}{%
\setlength{\leftmargin}{\labelwidth}
\addtolength{\leftmargin}{\parskip}
\setlength{\itemindent}{\listparindent}
\setlength{\itemsep}{\medskipamount}
\setlength{\topsep}{\itemsep}}
\setcounter{casectr}{0}
\usecounter{casectr}}
{\end{list}}
\newenvironment{elabeling}[2][]%
{\settowidth{\lyxlabelwidth}{#2}
\begin{description}[font=\normalfont,style=sameline,
leftmargin=\lyxlabelwidth,#1]}
{\end{description}}

\newcommand\pfun{\mathrel{\ooalign{\hfil\ensuremath{\mapstochar\mkern5mu}\hfil\cr\ensuremath{\to}\cr}}}
\usepackage[all]{hypcap}
\usepackage{stmaryrd}
\usepackage{bbm}
\usepackage{amsmath}
\newcommand\llparen\llparenthesis
\newcommand\rrparen\rrparenthesis
\newcommand\llbr{\llbracket}
\newcommand\rrbr{\rrbracket}
\def\twoheadleftrightarrowfill{\arrowfill@\twoheadleftarrow\relbar\twoheadrightarrow}
\newcommand{\undertwoheadleftrightarrow}{\mathpalette{\underarrow@\twoheadleftrightarrowfill}}

\usepackage{etoolbox,mathtools,placeins,mullength,breakurl}
\usepackage{mathbbol-nooverride}

\makeatother

\usepackage{babel}
 \providecommand{\interpretationname}{Interpretation}
\providecommand{\casename}{Case}
\providecommand{\corollaryname}{Corollary}
\providecommand{\definitionname}{Definition}
\providecommand{\lemmaname}{Lemma}
\providecommand{\propositionname}{Proposition}
\providecommand{\remarkname}{Remark}
\providecommand{\theoremname}{Theorem}

\begin{document}

\lyxeptcstitlerunning{Continuation calculus}{Continuation calculus}

\lyxeptcsauthorrunning{Bram Geron, Herman Geuvers}{Bram Geron\institute{Eindhoven University of Technology}\email{bgeron@gmail.com}
\and  Herman Geuvers\institute{Radboud University Nijmegen\\
Eindhoven University of Technology}\email{herman@cs.ru.nl}}

\lyxeptcsevent{COS 2013}
\maketitle
\begin{abstract}
Programs with control are usually modeled using lambda calculus extended
with control operators. Instead of modifying lambda calculus, we consider
a different model of computation. We introduce continuation calculus,
or CC, a deterministic model of computation that is evaluated using
only head reduction, and argue that it is suitable for modeling programs
with control. It is demonstrated how to define programs, specify them,
and prove them correct. This is shown in detail by presenting in CC
a list multiplication program that prematurely returns when it encounters
a zero. The correctness proof includes termination of the program.

In continuation calculus we can model both call-by-name and call-by-value.
In addition, call-by-name functions can be applied to call-by-value
results, and conversely.

\end{abstract}
\ifdef{\lyxeptcstitlerunning}{}{

\title{Continuation calculus\\
(omgezet naar plain \emph{article})}

\author{Bram Geron \and  Herman Geuvers}
\maketitle
\begin{abstract}
De abstract staat in ccalc-master.lyx.
\end{abstract}
}

\newcommand{\fieldampersand}{&}\global\long\def\twoheadleftrightarrow{\;\mbox{\ensuremath{\twoheadleftarrow}\!\!\!\!\ensuremath{\twoheadrightarrow}}\;}
\global\def \rawprotect{\protect}

\global\long\def\d#1#2{\mathit{#1}\longrightarrow\mathit{#2}}
\global\long\def\dsplit#1#2{\mathit{#1}\fieldampersand\longrightarrow\mathit{#2}}

\global\long\def\Names{\mathcal{N}}
\global\long\def\Vars{\mathcal{V}}
\global\long\def\Univ{\mathcal{U}}
\global\long\def\LHS{\mathsf{LHS}}
\global\long\def\RHS{\mathsf{RHS}}
\global\long\def\Rules{\mathsf{Rules}}
\global\long\def\Programs{\mathsf{Programs}}

\global\long\def\Bools{\mathbb{B}}
\global\long\def\Naturals{\mathbb{N}}
\global\long\def\ListNats{\mathconst{List}_{\mathbb{N}}}

\global\long\def\head#1{\mathsf{head}(#1)}
\global\long\def\plainnext{\mathsf{next}}
\global\long\def\next#1#2{\plainnext_{#1}(#2)}
\global\long\def\arity#1#2{\mathsf{arity}_{#1}(#2)}
\global\long\def\length#1{\mathsf{length}(#1)}

\ifdef{\lyxeptcstitlerunning}{\global\long\def\llangle{\langle\!\!\langle}
\global\long\def\rrangle{\rangle\!\!\rangle}
}{\global\long\def\llangle{\langle\negmedspace\langle}
\global\long\def\rrangle{\rangle\negmedspace\rangle}
}

\global\long\def\sembedraw#1#2{\langle#1\rangle_{#2}}
\global\long\def\embedsraw#1#2{\llangle#1\rrangle_{#2}}

\global\long\def\sembed#1{\langle#1\rangle}
\global\long\def\sembedm#1{\sembed{\mathsf{#1}}}
\global\long\def\embeds#1{\llangle#1\rrangle}
\global\long\def\embedsm#1{\embeds{\mathsf{#1}}}

\global\long\def\embed#1{XXX#1XXX}

\global\long\def\eqvterm#1{\underline{\:#1\rawprotect\vphantom{\embed{}}\:}}
\global\long\def\convergingwith#1{\{#1\}}

\global\long\def\interpraw#1#2{\llbr#1\rrbr_{#2}}
\global\long\def\Ninterp#1{\interpraw{#1}{\Naturals}}
\global\long\def\interp#1{\interpraw{#1}{}}
\global\long\def\cps#1{\llbr#1\rrbr}
\global\long\def\denot#1{\llbr#1\rrbr}
\global\long\def\straightembed#1#2{\mathsf{embed}_{#2}\left(#1\right)}
\global\long\def\powerset#1{\mathcal{P}(#1)}

\global\long\def\fin{\downarrow}
\global\long\def\evfinal{{\;\twoheadrightarrow\downarrow}}
\global\long\def\obseq{\approx}
\global\long\def\reaches{\twoheadrightarrow}
\global\long\def\backreaches{\twoheadleftarrow}
\global\long\def\otherconv#1{=_{#1}}
\global\long\def\conv{\otherconv P}

\global\long\def\CC{\mathrm{CC}}
\global\long\def\Aoperator{\mathcal{A}}
\global\long\def\Coperator{\mathcal{C}}
\global\long\def\Foperator{\mathcal{F}}
\global\long\def\Callcc{\mathsf{call/cc}}

\ifdef{\lyxeptcstitlerunning}{\global\long\def\place{{f\negmedspace\! r}}
}{\global\long\def\place{{f\negmedspace r}}
}

\global\long\def\otherplace{\mathfrak{n}}
\global\long\def\Domain{\mathcal{D}}
\global\long\def\Range{\mathcal{R}}

\global\long\def\Exampletype{\mathcal{T}}
\global\long\def\mathconst#1{\mathsf{#1}}
\global\long\def\concat{\ensuremath{+\!\!\!\!+\,}}
\global\long\def\dom#1{\textsf{dom}(#1)}

\global\long\def\And{,\:\:}
\global\long\def\If{\text{if }}
\global\long\def\Ex{\text{for some }}
\global\long\def\St{\text{s.t. }}
\global\long\def\With{\text{with }}

\global\long\def\Let{\textsf{let }}
\global\long\def\Letrec{\textsf{let rec }}
\global\long\def\In{\textsf{in }}
\global\long\def\Match#1{\textsf{match }#1\textsf{ with}}
\global\long\def\Caseof#1{\textsf{case }#1\textsf{ of}}
\global\long\def\Where{\textsf{where }}
\global\long\def\Product{\textsf{product }}
\global\long\def\insteps#1{\qquad\mbox{in #1 steps}}
\global\long\def\myqedhere{\hfill\mbox{\qedhere}}

\global\long\def\mvar#1{\underline{\mathit{#1}}}

\section{\label{sec:introduction}Introduction }

\global\long\def\Add{\mathit{Add}}
 \global\long\def\Zero{\mathit{Zero}}
 \global\long\def\Succ{\mathit{S}}
\global\long\def\AddCBV{\mathit{AddCBV}}

Lambda calculus has historically been the foundation of choice for
modeling programs in pure functional languages. To capture features
that are not purely functional, there is an abundance of variations
in syntax and semantics: lambda calculus can be extended with special
operators: $\Aoperator$, $\Coperator$, $\Foperator_{+/-}^{+/-}$,
$\#$, and $\Callcc$, to incorporate control~\cite{Dyvbig07,felleisen1986,felleisen1987,Felleisen88}
or one can move to a calculus like $\lambda\mu$~\cite{parigot1992,ariola2007}
that allows the encoding of control-like operators. Also, one must
choose between the call-by-value and call-by-name reduction orders
for the calculus to correspond to the modeled language. If one wants
to study these calculi, one usually applies one of many CPS translations~\cite{Plotkin75,Danvy92}
which allow simulation of control operators in a system without them.
There is also a close connection between proofs in classical logic
and control operators, as was first pointed out by~Griffin~\cite{griffin1990},
who extended the Curry-Howard proofs-as-programs principle to include
rules of classical logic. The $\lambda\mu$-calculus of~\cite{parigot1992,ariola2003}
is also based on the relation between classical logical rules and
control-like constructions in the type theory.

In this paper, we introduce a different kind of calculus for formalizing
functional programs: \emph{continuation calculus}. It is deterministic
and Turing complete, yet its operational semantics are minimal: there
is only head reduction and no stack, environment, or context. Control
is natural to express, without additional operators.

We present continuation calculus as an untyped system. The study of
a typed version, and possibly the connections with the rules of classical
logic, is for future research. In the present paper we want to introduce
the system, show how to write programs in it and prove properties
about these programs, and show how control aspects and call-by-value~(CBV)
and call-by-name~(CBN) naturally fit into the system.

\medskip{}

Continuation calculus looks a bit like term rewriting and a bit like
$\lambda$-calculus, and it has ideas from both. A term in CC is of
the shape
\[
n.t_{1}.\cdots.t_{k},
\]
where $n$ is a \emph{name} and the $t_{i}$ are themselves terms.
The ``dot'' is a binary operator that associates to the left. Note
that terms do not contain variables. A\emph{ program} $P$ is a list
of \emph{program rules} of the form
\[
\d{n.x_{1}.\cdots.x_{k}}u
\]
where the $x_{i}$ are all different variables and $u$ is a term
over variables $x_{1}\ldots x_{k}$. This program rule is said to
\emph{define} $n$, and we make sure that in a program $P$ there
is \emph{at most one} definition of $n$. Here, CC already deviates
from term rewriting, where one would have, for example: 
\[
\begin{aligned}\dsplit{\Add(0,m)}m\\
\dsplit{\Add(S(n),m)}{S(\Add(n,m))}
\end{aligned}
\]
These syntactic case distinctions, or \emph{pattern matchings}, are
not possible in CC.

The meaning of the program rule $\d{n.x_{1}.\cdots.x_{k}}u$ is that
a term $n.t_{1}.\cdots.t_{k}$ evaluates to $u[\vec{x}:=\vec{t}]$:
the variables $\vec{x}$ in $t$ are replaced by the respective terms
$\vec{t}$. A peculiarity of CC is that one cannot evaluate ``deep
in a term'': we do not evaluate inside any of the $t_{i}$ and if
we have a term $n.t_{1}.\cdots.t_{m}$, where $m>k$, this term does
not evaluate. (This will even turn out to be a ``meaningless'' term.)

To give a better idea of how CC works, we give the example of the
natural numbers: how they are represented in CC and how one can program
addition on them. A natural number is either $0$, or $S(m)$ for
$m$ a natural number. We shall have a name $\Zero$ and a name $\Succ$.
The number $m$ will be represented by $\Succ.(\cdots.(S.\Zero)\cdots)$,
with $m$ times $\Succ$. So the numbers $0$ and $3$ are represented
by the terms $\Zero$ and $\Succ.(\Succ.(\Succ.\Zero))$.

The only way to extract information from a natural $m$ is to ``transfer
control'' to that natural. Execution should continue in some code
$c_{1}$ when $m=0$, and execution should continue in different code
$c_{2}$ when $m=S(p)$. This becomes possible by postulating the
following rules for $\Zero$ and $\Succ$: 
\[
\begin{aligned}\dsplit{\Zero.z.s}z\\
\dsplit{\Succ.x.z.s}{s.x}
\end{aligned}
\]
We will now implement call-by-value addition in CC on these natural
numbers. The idea of CC is that a function application does not just
produce an output value, but passes it to the next function, the continuation.
So we are looking for a term $\AddCBV$ that behaves as follows:
\begin{equation}
\AddCBV.\embeds m.\embeds p.r\reaches r.\embeds{m+p}\label{eq:spec-addcbv}
\end{equation}
 for all $m,p,r$, where $\reaches$ is the multi-step evaluation,
and $\embeds l$ are the terms that represent a natural number~$l$.
Term $r$ indicates where evaluation should continue after the computation
of $\embeds{m+p}$.

Equation (\ref{eq:spec-addcbv}) is the \emph{specification of $\AddCBV$}.
We will use the following algorithm: 
\[
\begin{aligned}0+p & =p\\
\Succ(m)+p & =m+\Succ(p)
\end{aligned}
\]

To program $\AddCBV$, we have to give a rule of the shape $\d{\AddCBV.x.y.r}t$.
We need to make a case distinction on the first argument $x$. If
$x=\Zero$, then the result of the addition is $y$, so we pass control
to $r.y$. If $x=S.u$, then control should eventually transfer to
$r.\left(\AddCBV.u.(\Succ.y)\right)$. Let us write down a first approximation
of $\AddCBV$:
\[
\d{\AddCBV.x.y.r}{x.(r.y).t}
\]
The term $t$ is yet to be determined. Now control transfers to $r.y$
when $x=\mathit{Zero}$, or to $t.u$ when $x=\Succ.u$. From $t.u$,
control should eventually transfer to $\AddCBV.u.(\Succ.y).r$. Let
us write down a naive second approximation of $\Add$, in which we
introduce a helper name $B$. 
\[
\begin{aligned}\dsplit{\AddCBV.x.y.r}{x.(r.y).B}\\
\dsplit{B.u}{\AddCBV.u.(\Succ.y).r}
\end{aligned}
\]
Unfortunately, the second line is not a valid rule: $y$ and $r$
are variables in the right-hand side of $B$, but do not occur in
its left-hand side. We can fix this by replacing $B$ with $B.y.r$
in both rules. 
\[
\begin{aligned}\dsplit{\AddCBV.x.y.r}{x.(r.y).(B.y.r)}\\
\dsplit{B.y.r.u}{\AddCBV.u.(\Succ.y).r}
\end{aligned}
\]
This is a general procedure for representing data types and functions
over data in CC. We can now prove the correctness of $\AddCBV$ by
showing (simultaneously by induction on $m$) that 
\[
\begin{aligned}\mathit{AddCBV.\embeds m.\embeds p.r} & \reaches r.\embeds{m+p}\\
B.\embeds p.r.\embeds{m'} & \reaches r.\embeds{m'+p+1}
\end{aligned}
\]
 We formally define and characterize continuation calculus in the
following sections. In Section~\ref{sec:data}, we define the meaning
of $\embeds{\cdot}$, which allows us to give a specification for
call-by-name addition, $\mathit{AddCBN}$: 
\[
\mathit{AddCBN}.\embeds m.\embeds p\in\embeds{m+p}
\]
This statement means that $\mathit{AddCBN}.\embeds m.\embeds p$ is
equivalent to and compatible with $\mathit{S.(\cdots(S.Zero)\cdots)}$,
with $m+p$ times $S$. The precise meaning of this statement will
be given in Definitions~\ref{def:embeds-set} and Remark~\ref{rem:lifted-embeds}.

The terms $\mathit{AddCBV}$ and $\mathit{AddCBN}$ are of a different
kind. Nonetheless, we will see in Section~\ref{sub:data-functions}
how call-by-value and call-by-name functions can be used together.
We show additional examples with $\mathit{FibCBV}$ and $\mathit{FibCBN}$
in Section~\ref{sub:data-functions}. Furthermore, we model and prove
a program with $\Callcc$ in Sections~\ref{sec:listmult-modeling}
and~\ref{sec:listmult-correctness}.

The authors have made a program available to evaluate continuation
calculus terms on \url{http://www.cs.ru.nl/~herman/ccalc/}. Evaluation
traces of the examples are included.

\section{Formalization}
\begin{defn}[names]
There is an infinite set $\Names$ of \index{names}names. Concrete
names are typically denoted as upper-case letters ($\mathit{A,B},\ldots$),
or capitalized words ($\mathit{True,False,And,\ldots}$); we refer
to \emph{any} name using $n$ and $m$.\end{defn}
\begin{interpretation*}
Names are used by programs to refer to `functionality', and will serve
the role of constructors, function names, as well as labels within
a function.\end{interpretation*}
\begin{defn}[universe]
\index{universe Univ
@universe $\Univ$}The set of terms $\Univ$ in continuation calculus is generated by:
\[
\Univ::=\Names\,|\,\Univ.\Univ
\]
 where . (dot) is a binary constructor. The \index{dot}dot is neither
associative nor commutative, and there shall be no overlap between
names and dot-applications. We will often use $M,N,t,u$ to refer
to terms. If we know that a term is a name, we often use $n,m$. We
sometimes use lower-case words that describe its function, e.g.\ $\mathit{abort}$,
or letters, e.g.\ $r$ for a `return continuation'.

The dot is read left-associative: when we write $A.B.C$, we mean
$(A.B).C$.\end{defn}
\begin{interpretation*}
Terms by themselves do not denote any computation, nor do they have
any value of themselves. We inspect value terms by `dotting' other
terms on them, and observing the reduction behavior. If for instance
$b$ represents a boolean value, then $\mathit{b.t.f}$ reduces to
$t$ if $b$ represents $\mathconst{true}$; $\mathit{b.t.f}$ reduces
to $f$ if $b$ represents $\mathconst{false}$.\end{interpretation*}
\begin{defn}[head, length]
All terms have a \index{head of a term}head, which is defined inductively:
\[
\begin{aligned}\head{n\in\Names} & =n\\
\head{a.b} & =\head a.
\end{aligned}
\]
 The head of a term is always a name.

The \emph{\index{length of term}}length of a term is determined by
the number of dots traversed towards the head. 
\[
\begin{aligned}\length{n\in\Names} & =0\\
\length{a.b} & =1+\length a.
\end{aligned}
\]

\end{defn}
This definition corresponds to left-associativity: $\length{n.t_{1}.t_{2}.\cdots.t_{k}}=k$.
\begin{defn}[variables]
There is an infinite set $\mathcal{V}$ of \index{variables Vars
@variables $\Vars$}variables. Terms are not variables, nor is the result of a dot application
ever a variable. 

Variables are used in CC rules as formal parameters to refer to terms.
We will use lower-case letters or words, or $x,y,z$ to refer to variables. 
\end{defn}
Note that we use similar notations for both variables and terms. However,
variables exist only in rules, so we expect no confusion.

\begin{defn}[rules]
\index{rule}Rules consist of a \index{left-hand side}\index{LHS}left-hand
and a \index{right-hand side}\index{RHS}right-hand side, generated
by: 
\[
\begin{aligned}\LHS & ::=\Names\,|\:\LHS.\Vars\qquad\mbox{where every variable occurs at most once}\\
\RHS & ::=\Names\,|\,\Vars\,|\:\RHS.\RHS
\end{aligned}
\]

Therefore, any right-hand side without variables is a term in $\Univ$.

A combination of a left-hand and a right-hand side is a rule only
when all variables in the right-hand side also occur in the left-hand
side. 
\[
\Rules::=\LHS\rightarrow\RHS\qquad\mbox{where all variables in \ensuremath{\RHS}\ occur in \ensuremath{\LHS}}
\]

A rule is said to \index{defining rule}\emph{define} the name in
its left-hand side; this name is also called the \emph{\index{head of a rule}head}.
The \index{length of left-hand side}\emph{length} of a left-hand
side is equal to the number of variables in it. 
\end{defn}

\begin{defn}[program]
A\emph{ }\index{program}\emph{program} is a finite set of rules,
where no two rules define the same name. We denote a program by $P$.
\[
\Programs=\{P\subseteq\Rules|P\text{ is finite and }\head{\cdot}\text{ is injective on the LHSes in }P)
\]

The \index{domain}\emph{domain} of a program is the set of names
defined by its rules. 
\[
\dom P=\{\head{\mathit{rule}}\,|\,\mathit{rule}\in P\}
\]

We will frequently \index{extend}extend programs: an \index{extension}\emph{extension}
of a program is a superset of that program. 
\end{defn}

\begin{defn}[evaluation]
\label{def:evaluation}A term can be evaluated under a program. \index{evaluation}Evaluation
consists of zero or more sequential steps, which are all deterministic.
For some terms and programs, evaluation never terminates.

We define the evaluation through the partial successor function $\next P{\cdot}:\Univ\pfun\Univ$.
We define $\next Pt$ when $P$ defines $\head t$, and $\length t$
equals the length of the corresponding left-hand side. 
\[
\next P{n.t_{1}.t_{2}.\cdots.t_{k}}=r[\vec{x}:=\vec{t}]\qquad\mbox{when ``\ensuremath{\d{n.x_{1}.x_{2}.\cdots.x_{k}}r}''}\in P
\]

It is allowed that $n=0$: 
\[
\next Pn=r\qquad\mbox{when ``\ensuremath{\d nr}''}\in P
\]

More informally, we write $M\rightarrow_{P}N$ when $\next PM=N$.
The reflexive and transitive closure of $\rightarrow_{P}$ will be
denoted $\twoheadrightarrow_{P}$. When $M\reaches N$, then we call
$N$ a \index{reduct}\emph{reduct} of $M$, and $M$ is said to be
\emph{defined}. When $\next PM$ is not defined, we write that $M$
is \emph{final}. Notation: $M\fin_{P}$. We also combine the notations:
if $\next PM=N$ and $\next PN$ is undefined, we may write $M\rightarrow_{P}N\fin_{P}$.
We will often leave the subscript $P$ implicit: $M\rightarrow N\fin$.

In Section~\ref{sec:term-categorizations}, we divide the final terms
in three groups: undefined terms, incomplete terms, and invalid terms.
Thus, these are the three cases where $\next PM$ is undefined.
\end{defn}

\begin{defn}[termination]
A term $M$ is said to be \emph{terminating} under a program $P$,
notation $M\evfinal_{P}$, when it has a final reduct: $\exists N\in\Univ:M\reaches_{P}N\fin_{P}$.
We often leave the subscript $P$ implicit.
\end{defn}

\section{\label{sec:term-categorizations}Categorization of terms}

A program divides all terms into four disjoint categories: undefined,
incomplete, complete, and invalid. A term's evaluation behavior depends
on its category, to which the term's\emph{ arity} is crucial.
\begin{defn}
The name $n$ has \emph{arity} $k$ if $P$ contains a rule of the
form $\d{n.x_{1}.\cdots.x_{k}}q$.

A term $t$ has \emph{arity} $k-i$ if it is of the form $n.q_{1}.\cdots.q_{i}$,
where $n$ has arity $k$ ($k\ge i$).
\end{defn}

\begin{defn}
Term $t$ is \emph{defined in $P$} if $\head t\in\dom P$, otherwise
we say that $t$ is \emph{undefined}.

Given a $t$ that is defined, we say that
\begin{itemize}
\item $t$ is \emph{complete} if the arity of $t$ is 0
\item $t$ is \emph{incomplete} if the arity of $t$ is $j>0$
\item $t$ is \emph{invalid} if is has no arity (that is, $t$ is of the
form $n.q_{1}.\cdots.q_{i}$, where $n$ has arity $k<i$)
\end{itemize}
\end{defn}
The four categories have distinct characteristics.
\begin{description}
\item [{Undefined~terms.}]  Term $M$ is undefined iff $M.N$ is undefined.
Extension of the program causes undefined terms to remain undefined
or become incomplete, complete, or invalid.

\begin{interpretation*}
Because variables are not part of a term in continuation calculus,
we use undefined names instead for similar purposes, as exemplified
by Theorem~\ref{thm:freshness}. This means that all CC terms are
`closed' in the lambda calculus sense.
\end{interpretation*}

The remaining three categories contain \emph{defined terms}: terms
with a head $\in\dom P$. Extension of the program does not change
the category of defined terms.

\item [{Incomplete~terms.}]  If $M$ is incomplete, then $M.N$ can be
incomplete or complete. 

\begin{interpretation*}
There are four important classes of incomplete terms.
\begin{itemize}
\item \emph{Data terms (see Section~\ref{sec:data})}. If $d$ represents
$c_{k}(v_{1},\cdots,v_{n_{k}})$ of a data type $\mathcal{D}$ with
$m$ constructors, then $\forall t_{1}\ldots t_{m}:d.\vec{t}\reaches t_{k}.\vec{v}$.
Examples: 
\[
\hspace{-30bp}\begin{aligned}\forall z,s:\mathit{Zero.z.s} & \reaches z & \quad & \mathit{Zero}\mbox{ represents }0\\
\forall z,s:\mathit{S.(S.(S.Zero))}.z.s & \reaches s.(S.(S.Zero)) &  & \mathit{S.(S.(S.Zero))}\mbox{ represents }\Succ(\Succ(\Succ(0)))
\end{aligned}
\]

\item \emph{Call-by-name function terms}. These are terms $f$ such that
$f.v_{1}.\cdots.v_{k}$ is a data term $\in\denot{\mathcal{D}}$ for
all $\vec{v}$ in the appropriate domain. Example using Figure~\ref{fig:definition-fib-add-cbv-cbn}:
\[
\hspace{-30bp}\begin{aligned}\forall z,s:\mathit{AddCBN.Zero.Zero.z.s} & \reaches z\\
\forall z,s:\mathit{AddCBN}.(S.Zero).(S.(S.Zero)).z.s & \reaches s.(\mathit{AddCBN.Zero.(S.(S.Zero))})
\end{aligned}
\]

Recall that $\mathit{AddCBN.(S.Zero)}.(S.(S.Zero))$ is a data term
that represents $3$. The second reduction shows that $ $$1+_{\mathrm{CBN}}2=\Succ(x)$,
for some $x$ represented by $\mathit{AddCBN.Zero.(S.(S.Zero))}$.

\item \emph{Call-by-value function terms}. These are terms $f$ of arity
$n+1$ such that for all $\vec{v}$ in a certain domain, $\forall r:f.v_{1}.\cdots.v_{n}.r\reaches r.t$
with data term $t$ depending only on $\vec{v}$, not on $r$. Example:
\[
\forall r:\mathit{\AddCBV.(S.Zero).(S.(S.Zero)).r\reaches r.(S.(S.(S.Zero)))}
\]

\item \emph{Return continuations}. These represent the state of the program,
parameterized over some values. Imagine a C program fragment ``\texttt{return
abs(2 - ?);}''. If we were to resume execution from such fragment,
then the program would run to completion, but it is necessary to first
fill in the question mark. If $r$ represents the above program fragment,
then $r.3$ represents the completed fragment ``\texttt{return abs(2
- 3);}''.

If a return continuation has arity $n$, then it corresponds to a
program fragment with $n$ question marks.

\end{itemize}
\end{interpretation*}
\item [{Invalid~terms.}] All invalid terms will be considered equivalent.
If $M$ is invalid, then $M.N$ is also invalid.
\item [{Complete~terms.}] This is the set of terms that have a successor.
If $M$ is complete, then $M.N$ is invalid.
\end{description}
\begin{figure}
\begin{centering}
\begin{tabular}{c}
\textbf{Common definitions}\tabularnewline
$\begin{aligned}\dsplit{Zero.z.s}z\\
\dsplit{S.m.z.s}{s.m}\\
\dsplit{Nil.ifempty.iflist}{ifempty}\\
\dsplit{Cons.n.l.ifempty.iflist}{iflist.n.l}
\end{aligned}
$\tabularnewline
\end{tabular}
\par\end{centering}

\bigskip{}

\global\long\def\alignedvspace{\vphantom{\underset{\mbox{example}}{}}}

\begin{centering}
\hspace*{-20bp}%
\begin{tabular}{>{\centering}p{0.52\textwidth}|>{\centering}p{0.58\textwidth}}
\textbf{Call-by-value functions} & \textbf{Call-by-name functions}\tabularnewline
$\begin{aligned}[t]\dsplit{AddCBV.x.y.r}{x.(r.y).(AddCBV'.y.r)}\\
\dsplit{AddCBV'.y.r.x'}{AddCBV.x'.(S.y).r}\alignedvspace\\
\dsplit{FibCBV.x.r}{x.(r.Zero).(FibCBV_{1}.r)}\\
\dsplit{FibCBV_{1}.r.y}{y.(r.(S.Zero)).(FibCBV_{2}.r.y)}\\
\dsplit{FibCBV_{2}.r.y.y'}{FibCBV.y.(FibCBV_{3}.r.y')}\\
\dsplit{FibCBV_{3}.r.y'.fib_{y}}{FibCBV.y'.(FibCBV_{4}.r.fib_{y})}\\
\dsplit{FibCBV_{4}.r.fib_{y}.fib_{y'}}{AddCBV.fib_{y}.fib_{y'}.r}
\end{aligned}
$ & $\begin{aligned}[t]\dsplit{AddCBN.x.y.z.s}{x.(y.z.s).(AddCBN'.y.s)}\\
\dsplit{AddCBN'.y.s.x'}{s.(AddCBN.x'.y)}\alignedvspace\\
\dsplit{FibCBN.x.z.s}{x.z.(FibCBN_{1}.z.s)}\\
\dsplit{FibCBN_{1}.z.s.y}{y.(s.Zero).(FibCBN_{2}.z.s.y)}\\
\dsplit{FibCBN_{2}.z.s.y.y'}{AddCBN.(FibCBN.y).(FibCBN.y').z.s}
\end{aligned}
$\tabularnewline
\end{tabular}
\par\end{centering}

\caption{\label{fig:definition-fib-add-cbv-cbn}Continuation calculus representations
of $+$ and $\mathconst{fib}$. The functions are applied in a different
way, as shown in Figure~\ref{fig:using-add-plus-cbv-cbn}. This incompatibility
is already indicated by the different arity: $\arity{}{\mathit{AddCBV}}=3\ne\arity{}{\mathit{AddCBN}}=4$,
and $\arity{}{FibCBV}=2\ne\arity{}{\mathit{FibCBN}}=3$. Figure~\ref{fig:using-add-plus-cbv-cbn}
shows how to use the four functions. }
\end{figure}

\section{\label{sec:reasoning}Reasoning with CC terms}

This section sketches the nature of continuation calculus through
theorems. All proofs are included in the appendix.

\subsection{\label{sub:theorems-general}Fresh names}
\begin{defn}
When a name $\place$ does not occur in the program under consideration,
then we call $\place$ a \emph{fresh name}. Furthermore, all fresh
names that we assume within theorems, lemmas, and propositions are
understood to be different. When we say $\place$ is fresh \emph{for
some objects}, then it is additionally required that $\place$ is
not mentioned in those objects.
\end{defn}
We can always assume another fresh name, because programs are finite
and there are infinitely many names.
\begin{interpretation*}
Fresh names allow us to reason on arbitrary terms, much like free
variables in lambda calculus.\end{interpretation*}
\begin{thm}
\label{thm:freshness}Let $M,N$ be terms, and let name $\place$
be fresh. The following equivalences hold:
\[
\begin{aligned}M\reaches N & \Longleftrightarrow\forall t\in\Univ:M[\place:=t]\reaches N[\place:=t]\\
M\evfinal & \Longleftrightarrow\exists t\in\Univ:M[\place:=t]\evfinal
\end{aligned}
\]

\end{thm}

\begin{lem}[determinism]
\label{lem:determinism}Let $M,\vec{t},\vec{\vphantom{t}u}$ be terms,
and let $m,n$ be undefined names in $P$. If $M\twoheadrightarrow_{P}m.t_{1}.\cdots.t_{k}$
and $M\twoheadrightarrow_{P}n.u_{1}.\cdots.u_{l}$, then $m.t_{1}.\cdots.t_{k}=n.u_{1}.\cdots.u_{l}$.\end{lem}
\begin{rem}
If $m$ or $n$ is defined, this may not hold. For instance, in the
program ``$\d AB;\d BC$'', we have $A\reaches B$ and $A\reaches C$,
yet $B\ne C$.
\end{rem}

\subsection{\label{sub:theorems-equivalence}Term equivalence}

Besides syntactic equality ($=$), we introduce two equivalences on
terms: common reduct ($\conv$) and observational equivalence ($\obseq_{P}$).
\begin{defn}
Terms $M,N$ have a common reduct if $M\reaches t\backreaches N$
for some term $t$. Notation: $M\conv N$.\end{defn}
\begin{prop}
\label{prp:converging-final-then-reaches}Suppose $M\conv N\fin$.
Then $M\reaches N$. 
\end{prop}
Common reduct is a strong equivalence, comparable to $\beta$-conversion
for lambda calculus. Terms $M\ne N$ can only have a common reduct
if at least one of them is complete. This makes pure $\conv$ unsuitable
for relating data or function terms, which are incomplete. In fact,
$\conv$ is not a congruence. 

To remedy this, we define an observational equivalence in terms of
termination.
\begin{defn}
Terms $M$ and $N$ are \emph{observationally equivalent} under a
program $P$, notation $M\obseq_{P}N$, when for all extension programs
$P'\supseteq P$ and terms $X$: 
\[
X.M\evfinal_{P'}\Longleftrightarrow X.N\evfinal_{P'}
\]
We may write $M\obseq N$ if the program is implicit.
\end{defn}
Examples: $\mathit{AddCBV}.\sembed m.\sembed 0\obseq\mathit{AddCBV}.\sembed 0.\sembed m$
and $\sembed 0\obseq\sembedm{True}$, but $\sembed 0\not\obseq\sembed 1$;
see Section~\ref{sec:data}.
\begin{lem}
\label{lem:equiv-congruence}$\obseq$ is a congruence. In other words,
if $M\obseq M'$ and $N\obseq N'$, then $M.N\obseq M'.N'$.
\end{lem}

\paragraph{Characterization}

The reduction behavior of complete terms divides them in three classes.
Observational equivalence distinguishes the classes.
\begin{itemize}
\item \emph{Nontermination}. When $M$ is nonterminating and the program
is extended,~$M$ remains nonterminating.

If the reduction path of $M$ is finite, we call it \emph{terminating},
and we may write $M\evfinal$. This is shorthand for $\exists N\in\Univ:M\reaches N\fin$.

\item \emph{Proper reduction to an incomplete or invalid term}. All such
$M$ are observationally equivalent to an invalid term. When the program
is extended, such terms remain in their execution class.
\item \emph{Proper reduction to an undefined term}. Observational equivalence
distinguishes terms $M,N$ if the head of their final term is different.
Therefore, there are infinitely many subclasses.

When the program is extended, the final term may become defined. This
can cause such $M$ to fall in a different class. 

\end{itemize}
The following proposition and theorem show that $\obseq$ distinguishes
three equivalence classes: if $M\obseq N$, then $M$ and $N$ are
in the same class.
\begin{prop}
\label{prp:equiv-then-same-termination}If $M\obseq N$, then $M\evfinal\Leftrightarrow N\evfinal$.
\end{prop}

\begin{thm}
\label{thm:obseq-then-reaches-same-name}Let $M\obseq N$ and $M\reaches\place.t_{1}.\cdots.t_{k}\fin$
with $\place\notin\dom P$. Then $N\reaches\place.u_{1}.\cdots.u_{k}\fin$
for some $\vec{u}$.

\end{thm}

\paragraph{Retrieving observational equivalence}

Complete terms with a common reduct are observationally equal. If
$M,N$ are incomplete, but they have common reducts when extended
with terms, then also $M\obseq N$.
\begin{thm}
\label{thm:converging-ext-then-equiv}Let $M,N$ be terms with arity
$k$. If $M.t_{1}.\cdots.t_{k}\conv N.t_{1}.\cdots.t_{k}$ for all
$\vec{t}$, then $M\obseq N$.\end{thm}
\begin{cor}
\label{cor:converging-then-equiv}Suppose $M\conv N$ and $\arity{}M=\arity{}N=0$.
Then $M\obseq N$.
\end{cor}
\global\long\def\vecplace{\mathrlap{\;\vec{\phantom{f}}}\place}

\begin{rem}
\label{rem:goto-omega-to-but-not-obseq}$M\reaches N$ does not always
imply $M\obseq N$ if $\arity{}N>0$. For instance, take the following
program: 
\[
\begin{aligned}\dsplit{Goto.x}x\\
\dsplit{Omega.x}{x.x}
\end{aligned}
\]

Then $\mathit{Goto.Omega}\rightarrow\mathit{Omega}$, an incomplete
term. We cannot `fix' $\mathit{Goto.Omega}$ by appending another
term: $\mathit{Goto.Omega.Omega}$ is invalid. Name $\mathit{Goto}$
is defined for one `operand' term, and the superfluous $\mathit{Omega}$
term cannot be `memorized' as with lambda calculus. On the other hand,
$\mathit{Omega.Omega}\rightarrow\mathit{Omega.Omega}$ is nonterminating.
Hence, $\mathit{Goto.Omega}\rightarrow\mathit{Omega}$ but $\mathit{Goto.Omega}\not\obseq\mathit{Omega}$.
\end{rem}
Note that $\mathit{Goto.Omega}\not\obseq\mathit{Omega}$ is only possible
because $\arity{}{\mathit{Goto.Omega}}\ne\arity{}{\mathit{Omega}}$.

\subsection{\label{sub:theorems-program-substitution-union}Program substitution
and union}
\begin{defn}[fresh substitution]
Let $n_{1}\ldots n_{k}$ be names, and $m_{1}\ldots m_{k}$ be fresh
for $M$, all different. Then $M[\vec{n}:=\vec{m}]$ is equal to $M$
where all occurrences of $\vec{n}$ are simultaneously replaced by
$\vec{m}$, respectively. The \emph{fresh substitution} $P[\vec{n}:=\vec{m}]$
replaces all $\vec{n}$ by $\vec{m}$ in both left and right hand
sides of the rules of $P$.
\end{defn}
We can combine two programs by applying a fresh substitution to one
of them, and taking the union. As the following theorems shows, this
preserves most interesting properties.
\begin{thm}
\label{thm:program-transformation}Suppose that $P'\supseteq P$ is
an extension program, and $M,N$ are terms. Then the left hand equations
hold. Let $\sigma$ denote a fresh substitution $[\vec{n}:=\vec{m}]$.
Then the right hand equations hold. 
\[
\begin{array}{rclcrcl}
M\rightarrow_{P}N & \Longrightarrow & M\rightarrow_{P'}N & \qquad\qquad & M\rightarrow_{P}N & \Longleftrightarrow & M\sigma\rightarrow_{P\sigma}N\sigma\\
M\twoheadrightarrow_{P}N & \Longrightarrow & M\twoheadrightarrow_{P'}N &  & M\twoheadrightarrow_{P}N & \Longleftrightarrow & M\sigma\twoheadrightarrow_{P\sigma}N\sigma\\
M\fin_{P} & \Longleftarrow & M\fin_{P'}\: &  & M\fin_{P} & \Longleftrightarrow & M\sigma\fin_{P\sigma}\\
M\conv N & \Longrightarrow & M\otherconv{P'}N &  & M\conv N & \Longleftrightarrow & M\sigma\otherconv{P\sigma}N\sigma\\
M\obseq_{P}N & \Longrightarrow & M\obseq_{P'}N &  & M\obseq_{P}N & \Longleftrightarrow & M\sigma\obseq_{P\sigma}N\sigma
\end{array}
\]
\end{thm}
\begin{rem}
Names $\vec{n}$ are not mentioned in $M\sigma$ and $P\sigma$, so
we can apply Theorem~\ref{thm:program-transformation} with $\sigma^{-1}$
on $M\sigma$ and $P\sigma$. 
\end{rem}

\begin{thm}
\label{thm:equiv-program-shortening-lengthening}Suppose that $P'$
extends $P$, but $\dom{P'\setminus P}$ are not mentioned in $M$,
$N$, or $P$. \\
 Then $M\obseq_{P}N\Longleftrightarrow M\obseq_{P'}N$.
\end{thm}

\section{\label{sec:data}Data terms and functions}

In this section, we show how to program some standard data in continuation
calculus. We first give a canonical representation of data as CC terms.
We then give essential semantic characteristics, and show that other
terms have those characteristics as well. Observational equivalence
guarantees that termination of the whole program is only dependent
on those characteristics. In fact, it will prove possible to implement
``call-by-name values'', which delay computation until it is needed,
by relying on those characteristics.

\paragraph{Standard representation of data}

In Section \ref{sec:introduction}, we postulated terms for natural
numbers in continuation calculus. We will now give this standard representation
formally, as well as the representation of booleans and natural lists.
\begin{defn}
\label{def:sembed}For a mathematical object $o$, we define a standard
representation~$\sembed o$ of that object as a CC term, which we
call a \emph{data term}. We postulate that the rules in the middle
column are included in programs that use the corresponding terms.
\[
\begin{aligned}\sembedm{True} & =\mathit{True} & \dsplit{True.t.f}t &  &  & \hspace{-35bp} & \smash{\begin{aligned}\phantom{\sembedm{True}}\\
\left.\begin{aligned}\phantom{\sembedm{True}}\\
\phantom{\sembedm{True}}
\end{aligned}
\right\} \text{booleans}
\end{aligned}
}\\
\sembedm{False} & =\mathit{False} & \dsplit{False.t.f}f\\
\sembed 0 & =\mathit{Zero} & \dsplit{Zero.z.s}z &  &  & \hspace{-35bp} & \smash{\begin{aligned}\phantom{\sembedm{True}}\\
\left.\begin{aligned}\phantom{\sembedm{True}}\\
\phantom{\sembedm{True}}
\end{aligned}
\right\} \text{naturals}
\end{aligned}
}\\
\sembed{m+1} & =\mathit{S}.\sembed m & \dsplit{S.x.z.s}{s.x}\\
\sembed{[]} & =\mathit{Nil} & \dsplit{Nil.e.c}e &  &  & \hspace{-35bp} & \smash{\begin{aligned}\phantom{\sembedm{True}}\\
\left.\begin{aligned}\phantom{\sembedm{True}}\\
\phantom{\sembedm{True}}
\end{aligned}
\right\} \text{lists of naturals}
\end{aligned}
}\\
\sembed{m:l} & =\mathit{Cons}.\sembed m.\sembed l\quad & \dsplit{Cons.x.xs.e.c}{c.x.xs}
\end{aligned}
\]
\end{defn}
\begin{thm}
$\sembedm{True}\not\obseq\sembedm{False}$.\end{thm}
\begin{proof}
Observe that for all $t,f$, $\mathit{\sembedm{True}.t.f\to t}$ and
$\sembedm{False}.t.f\to f$. Take two fresh names $t$ and $f$. Contraposition
of Theorem~\ref{thm:obseq-then-reaches-same-name} proves $\sembedm{True}.t.f\not\obseq\sembedm{False}.t.f$.
Because $\obseq$ is a congruence with respect to dot, we can conclude
$\sembedm{True}\not\obseq\sembedm{False}$.
\end{proof}
Similar results hold for $\Naturals$ and $\ListNats$, but we do
not provide a proof here.

\paragraph{A broader definition}

The behavioral essence of these data terms is that they take a continuation
for each constructor, and they continue execution in the respective
continuation, augmented with the constructor arguments. For instance,
$\sembed 0.z.s\reaches z$ and $\sembed{n+1}.z.s\reaches s.\sembed n$.
 We can capture this essence in the following term sets; $\denot{\Naturals}$
and $\denot{\ListNats}$ are the smallest sets satisfying the following
equalities.
\[
\begin{aligned}\denot{\mathbb{B}}=\{M\in\Univ| & \:\forall t,f\in\Univ:M.t.f\reaches t\vee M.t.f\reaches f\}\\
\denot{\mathbb{N}}=\{M\in\Univ| & \left(\forall z,s\in\Univ:M.z.s\reaches z\right)\\
 & \vee\exists x\in\denot{\mathbb{N}}\;\forall z,s\in\Univ:M.z.s\reaches s.x\}\\
\denot{\mathconst{List}_{\mathbb{N}}}=\{M\in\Univ| & \left(\forall e,c\in\Univ:M.e.c\reaches e\right)\\
 & \vee\exists x\in\denot{\mathbb{N}},xs\in\denot{\mathconst{List}_{\mathbb{N}}}\;\forall e,c\in\Univ:M.e.c\reaches c.x.xs\}
\end{aligned}
\]

\begin{rem}
These sets are dependent on the program. The sets are monotone with
respect to program extension: if $M\in\denot{\Bools}$, $\in\denot{\Naturals}$,
or $\in\denot{\ListNats}$ for a program, then $M$ is also in the
corresponding set for any extension program.

The sets include other terms besides $\sembedm{True}$, $\sembedm{False}$,
$\sembed n$, and $\sembed l$. Consider the following program fragment,
which implements the $\le$ operator on natural numbers.
\[
\begin{aligned}\dsplit{Leq.x.y.t.f}{x.t.(Leq'.y.t.f)}\\
\dsplit{Leq'.y.t.f.x'}{Leq.y.x'.f.t}
\end{aligned}
\]

Given naturals $m,p$ and this program fragment, $\mathit{Leq.\sembed m.\sembed p}\in\denot{\Bools}$.
Even more, $\mathit{Leq.\sembed m.\sembed p}\obseq\sembed{m\le p}$.
In general, it follows from Theorem~\ref{thm:converging-ext-then-equiv}
that all $M\in\denot{\Bools}$ are observationally equivalent to $\sembedm{True}$
or $\sembedm{False}$. The appendix contains a proof of the analogous
statement for $\denot{\Naturals}$:\end{rem}
\begin{prop}
\label{prp:all-nats-eqv-succs-zero}All terms in $\denot{\mathbb{N}}$
are observationally equivalent to $\sembed k$ for some~$k$.
\end{prop}
For further reasoning, it is useful to split up $\denot{\Naturals}$
in parts as follows.
\begin{defn}
\label{def:embeds-set}For a natural number $k$, the set $\embeds k$
is defined as $\{M\in\denot{\Naturals}|M\obseq\sembed k\}$. We define
$\embeds b$ and $\embeds l$ analogously for booleans $b$ and lists
of naturals $l$.
\end{defn}
With this definition, we may say $\mathit{Leq}.\sembed 3.\sembed 4\in\embedsm{True}$.
In fact, if $a\in\embeds 3$ and $b\in\embeds 4$, then $\mathit{Leq}.a.b\in\embedsm{True}$.%
\footnote{To see this, observe $\mathit{Leq}.a.b\obseq\mathit{Leq}.\sembed 3.\sembed 4\obseq\sembedm{True}$
by congruence, then use Theorems~\ref{thm:obseq-then-reaches-same-name}
and~\ref{thm:freshness}.%
} To support this pattern of reasoning, we allow to lift $\embeds{\cdot}$,
denoting a \emph{term}. The resulting statements are implicitly quantified
universally and existentially, and are usable in proof chains.
\begin{rem}
\label{rem:lifted-embeds}For data terms, we would like to reason
and compute with equivalence classes of representations, $\embeds k$,
instead of with the representations themselves, $\sembed k$. Of course,
a CC program will always compute with a term (and not with an equivalence
class of terms), but we would like this computation to only depend
on the characterization of the equivalence class. 

For example, we want to compute a CBN addition function $\mathit{AddCBN}$,
such that for all $m,p\in\Naturals$, $\forall t\in\embeds m\forall u\in\embeds p:\mathit{AddCBN}.t.u\in\embeds{m+p}$.
As a specification, we want to summarize this as: 
\[
\mathit{AddCBN}.\embeds m.\embeds p\in\embeds{m+p}
\]

\end{rem}
We will also summarize a statement of the form $\forall t_{1}\in\embeds m\;\exists t_{2}\in\embeds m\;\exists t_{3}\in\embeds l:A.t_{1}\reaches B.t_{2}.t_{3}$
with the shorthand $A.\embeds m\reaches B.\embeds m.\embeds l$. If
we know $A.\embeds m\reaches B.\embeds m.\embeds l$ and $B.\embeds m.\embeds l\reaches C.\embeds m$,
then we may logically conclude 
\[
\forall t_{1}\in\embeds m\;\exists t_{2}\in\embeds m\;\exists t_{3}\in\embeds l\;\exists t_{4}\in\embeds m:A.t_{1}\reaches B.t_{2}.t_{3}\reaches C.t_{4}\mbox{ ,}
\]
which we will summarize as $A.\embeds m\reaches B.\embeds m.\embeds l\reaches C.\embeds m$.
Analogous statements of this form, and longer series, will be summarized
in a similar way. In particular, it will suit us to also use $\rightarrow$
and $\conv$ in longer derivations.

\emph{}

\paragraph{Example: delayed addition}

We will program a different addition on natural numbers: one that
delays work as long as possible, like in call-by-name programming
languages. We use the following algorithm, for natural numbers $m,p$:
\[
\begin{aligned}0+p & =p\\
\Succ(m)+p & =\Succ(m+p)
\end{aligned}
\]

The resulting name $\mathit{AddCBN}$ will be a `call-by-name' function,
with specification $\mathit{AddCBN}.\embeds m.\embeds p\in\embeds{m+p}$,
so we have to build a rule for $\mathit{AddCBN}$. Because $\mathit{AddCBN}.\embeds m.\embeds p\in\denot{\Naturals}$,
$\arity{}{\mathit{AddCBN}}=4$. We reduce the specification with a
case distinction on the first argument. 
\[
\mathit{\begin{aligned}AddCBN.\embeds 0.\embeds p.z.s & \conv\embeds p.z.s, & \quad & \mbox{(}\mathit{AddCBN}.\embeds 0.\embeds p\mbox{ has the same specification as }\embeds p\mbox{)}\\
\mathit{AddCBN}.\embeds{\Succ(m)}.\embeds p.z.s & \reaches s.\embeds{m+p}
\end{aligned}
}
\]

We must make the case distinction by using the specified behavior
of the first argument. This suggests a rule of the form $\d{AddCBN.x.y.z.s}{x.(y.z.s).(s.(AddCBN.x'.y))}$.
It almost works:
\[
\begin{aligned}\mathit{AddCBN}.\embeds 0.\embeds p.z.s & \reaches\embeds p.z.s\\
\mathit{AddCBN}.\embeds{S(m)}.\embeds p.z.s & \reaches s.(AddCBN.x'.\embeds p).\embeds m
\end{aligned}
\]

However, variable $x'$ is not in the left-hand side, so this is not
a valid rule. Furthermore, if $x=\Succ(x')$, then $x'$ would be
erroneously appended to $\mathit{s.(AddCBN.x'.y)}$. We fix $\mathit{AddCBN}$
with a helper name $\mathit{AddCBN'}$, with specification $\mathit{AddCBN'}.\embeds p.s.\embeds m\reaches s.\embeds{m+p}$.
\[
\begin{aligned}\dsplit{AddCBN.x.y.z.s}{x.(y.z.s).(AddCBN'.y.s))}\\
\dsplit{AddCBN'.y.s.x'}{s.(AddCBN.x'.y)}
\end{aligned}
\]

This version conforms to the specification. 
\[
\begin{aligned}\mathit{AddCBN}.\embeds 0.\embeds p.z.s & \reaches\embeds p.z.s=\embeds p\\
\mathit{AddCBN}.\embeds{S(m)}.\embeds p.z.s & \reaches\mathit{AddCBN'}.\embeds p.s.\embeds m\\
 & \rightarrow s.(\mathit{AddCBN}.\embeds m.\embeds p)=s.\embeds{m+p}
\end{aligned}
\]

\emph{}

\begin{figure}
\begin{centering}
\begin{tabular}{>{\raggedright}p{0.48\textwidth}|>{\raggedright}p{0.48\textwidth}}
\centering{}\textbf{Call-by-value $\mathconst{fib}(7)$} & \centering{}\textbf{Call-by-name $\mathconst{fib}(7)$}\tabularnewline
\hline 
\vspace{-5bp}To apply $f$ to $\vec{x}$, evaluate $f.\vec{x}.r\reaches r.y$
for some $r$. Then $y$ is the result.  & \vspace{-5bp}To apply $f$ to $\vec{x}$, write $f.\vec{x}$. This
is directly a data term, no reduction happens.\tabularnewline
\vspace{-5bp}\emph{The result of $\mathconst{fib}(7)$ is 13, obtained
in 362 reduction steps:}

$\mathit{FibCBV}.7.\place\reaches\place.13$ & \vspace{-5bp}\emph{By the specification of $\mathit{FibCBN}$, we
know $\mathit{FibCBN}.7\in\embeds{13}$. No reduction is involved.}\tabularnewline
\multicolumn{2}{>{\centering}p{1\textwidth}}{\vspace{-2bp}Both $13$ and $\mathit{FibCBN}.7$ can be used in other
functions, like $+$. Because they are observationally equivalent,
they can be substituted for each other in a term. That does not affect
termination, or the head of the final term if that is undefined (Theorem~\ref{thm:obseq-then-reaches-same-name}).
However, substituting $13$ for $\mathit{FibCBN}.7$ may make the
evaluation shorter. \vspace{8bp}}\tabularnewline
\emph{$13+_{\text{CBV}}0$ is obtained in 41 steps:} & \emph{$\mathit{FibCBN}.7+_{\text{CBV}}0$ is obtained in 304 steps:}\tabularnewline
$\mathit{AddCBV}.13.0.\place\reaches\place.13\insteps{41}$ & $\mathit{AddCBV}.(\mathit{FibCBN}.7).0.\place\reaches\mathit{\place.13}\insteps{304}\mbox{ (263 more)}$\vspace{8bp}\tabularnewline
\multicolumn{2}{c}{Our implementation of $\mathit{AddCBV}$ does not examine the right
argument, as the converse addition shows.\vspace{8bp}}\tabularnewline
$\mathit{AddCBV.0.13.\place}\reaches\mathit{\place.13}\insteps 2$ & $\mathit{AddCBV.0.(FibCBN.7).\place}\reaches\mathit{\place.(FibCBN.7)}\insteps 2$\tabularnewline
\end{tabular}
\par\end{centering}

\caption{\label{fig:using-add-plus-cbv-cbn}Calculating $\mathconst{fib}(7)$,
$\mathconst{fib}(7)+0$, and $0+\mathconst{fib}(7)$ using call-by-value
and call-by-name. Effectively, $\mathit{FibCBN}$ delays computation
until it is needed. A natural number $n$ stands for $\smash{\protect\underbrace{S.(\cdots.(S}_{n\text{ times}}}.\mathit{Zero})\cdots)$.}
\end{figure}

\subsection{\label{sub:data-functions}Call-by-name and call-by-value functions}

We regard two kinds of functions. We call them \emph{call-by-name}
and \emph{call-by-value}, by analogy with the evaluation stategies
for lambda calculus. Figure~\ref{fig:definition-fib-add-cbv-cbn}
defines a CBN and CBV version of addition on naturals and the Fibonacci
function. Figure~\ref{fig:using-add-plus-cbv-cbn} shows how to use
them. It also illustrates that the CBV function performs work eagerly,
while the CBN function delays work until it is needed: hence the analogy.
\begin{itemize}
\item \emph{Call-by-name functions} are terms $f$ such that $f.v_{1}.\cdots.v_{k}$
is a data term for all $\vec{v}$ in a certain domain. Example specifications
for such $f$: 
\[
\begin{aligned}\mathit{AddCBN}.\embeds m.\embeds p & \in\embeds{m+p}\\
\mathit{FibCBN}.\embeds m.\embeds p & \in\embeds{\mathconst{fib}(m)}
\end{aligned}
\]

\item \emph{Call-by-value functions} are terms $f$ of arity $n+1$ such
that for all $\vec{v}$ in a certain domain, \\
$\forall r:f.v_{1}.\cdots.v_{n}.r\reaches r.t$ with data term $t$
depending only on $\vec{v}$, not on $r$. Example specifications
for such $f$: 
\[
\begin{aligned}\forall r:\mathit{AddCBV}.\embeds m.\embeds p.r & \reaches r.\embeds{m+p}\\
\forall r:\mathit{FibCBV}.\embeds m.r & \reaches r.\sembed{\mathconst{fib}(m)}
\end{aligned}
\]

The output of $\mathit{FibCBV}$ is always a standard representation.
 Because our implementation of $\mathit{AddCBV}$ does not inspect
the second argument, its output may not be a standard integer. An
example of this is shown in Figure~\ref{fig:using-add-plus-cbv-cbn}.

\end{itemize}
We leave formal proofs of the specifications for future work. 

\FloatBarrier

\section{\label{sec:listmult-modeling}Modeling programs with control}

To illustrate how control is fundamental to continuation calculus,
we give an example program that multiplies a list of natural numbers,
and show how an escape from a loop can be modeled without a special
operator in the natural CC representation. We use an ML-like programming
language for this example, and show the corresponding call-by-value
program for CC.

The naive way to compute the product of a list is as follows: 

\vspace{\abovedisplayskip}%
\begin{tabular}{>{\raggedright}p{0.4\textwidth}|>{\raggedleft}p{0.5\textwidth}}
$\begin{array}{l}
\Letrec\mathit{listmult}_{1}\; l=\Match l\\
\quad|\;[]\rightarrow1\\
\quad|\;\left(x:xs\right)\rightarrow x\cdot\mathit{listmult}_{1}\; xs
\end{array}$ & $\begin{aligned}\dsplit{ListMult.l.r}{l.(r.(S.Zero)).(C.r)}\\
\dsplit{C.r.x.xs}{ListMult.xs.(PostMult.x.r)}\\
\dsplit{PostMult.x.r.y}{Mult.x.y.r}
\end{aligned}
$\tabularnewline
\end{tabular}\vspace{\belowdisplayskip}

\noindent Note that if $l$ contains a zero, then the result is always
zero. One might wish for a more efficient version that skips all numbers
after zero.

\vspace{\abovedisplayskip}%
\begin{tabular}{>{\raggedright}p{0.4\textwidth}|>{\raggedleft}p{0.5\textwidth}}
$\begin{array}{l}
\Letrec\mathit{listmult}_{2}\; l=\Match l\\
\quad|\;[]\rightarrow1\\
\quad|\;\left(x:xs\right)\rightarrow\Match x\\
\quad\qquad|\;0\rightarrow0\\
\quad\qquad|\; x'+1\rightarrow x\cdot\mathit{listmult}_{2}\; l
\end{array}$ & $\begin{aligned}\dsplit{ListMult.l.r}{l.(r.(S.Zero)).(B.r)}\\
\dsplit{B.r.x.xs}{x.(r.Zero).(C.r.x.xs)}\\
\dsplit{C.r.x.xs.x'}{ListMult.xs.(PostMult.x.r)}\\
\dsplit{PostMult.x.r.y}{Mult.x.y.r}
\end{aligned}
$\tabularnewline
\end{tabular}\vspace{\belowdisplayskip}

\noindent However, $\mathit{listmult}_{2}$ is not so efficient either:
if the list is of the form $[x_{1}+1,\cdots,x_{k}+1,0]$, then we
only avoid multiplying $0\cdot\mathit{listmult}_{2}\,[]$. The other
multiplications are all of the form $n\cdot0=0$. We also want to
avoid execution of those surrounding multiplications. We can do so
if we extend ML with the $\Callcc$ operator, which creates alternative
exit points that are invokable as a function.

\vspace{\abovedisplayskip}%
\begin{tabular}{cc>{\raggedright}m{9em}}
$\begin{array}{l}
\Let\mathit{listmult}_{3}\: l=\\
\quad\Callcc\;\big(\lambda\mathit{abort}.\\
\qquad A\; l\\
\qquad\Where A=\mathsf{function}\\
\qquad\quad|\;[]\rightarrow1\\
\qquad\quad|\;\left(x:xs\right)\rightarrow\fbox{\negmedspace\negmedspace\ensuremath{\begin{array}[t]{l}
 \Match x\\
\quad|\;0\rightarrow\mathit{abort}\;0\\
\quad|\;\ x'+1\rightarrow\fbox{\ensuremath{x\cdot A\ xs}}_{(C)} 
\end{array}}\negmedspace\negmedspace}_{\mathrlap{(B)}}\:)
\end{array}$ & \qquad{} & \emph{\footnotesize The boxes are not syntax, but are used to relate
$\mathit{listmult}{}_{3}$ to Figure~\ref{fig:listmult}.}\emph{ }\tabularnewline
\end{tabular}\vspace{\belowdisplayskip}

\medskip{}

While $\mathit{listmult}_{3}$ is not readily expressible in actual
ML or lambda calculus, it is natural to express in~CC: we list the
program in Figure~\ref{fig:listmult}.

These programs are a CPS translation of $\mathit{listmult}_{3}$,
with one exception: the variable $\mathit{abort}$ in Figure~\ref{fig:listmult}
corresponds to the partial application of $\mathit{abort}$ to 0 in
$\mathit{listmult}_{3}$. Note that in CC, $\mathit{abort}$ is obtained
simply by constructing $\mathit{r.Zero}$. The variable $r$ globally
corresponds to the return continuation that is implicit in ML. Continuation
calculus requires to explicitly thread variables through the continuations.

\begin{figure}[tbh]

\noindent \begin{centering}
\makebox[1\textwidth]{%
\global\long\def\listmultskip{\vspace{7bp}}
\begin{tabular}[t]{>{\raggedright}p{0.55\linewidth}>{\raggedright}p{0.45\linewidth}}
\textbf{Continuation calculus} & \textbf{Haskell equivalent}\listmultskip\tabularnewline
\textemdash{} \emph{Assume} $m,m',p\in\mathbb{N}$, $l\in\mathconst{List}_{\mathbb{N}}$,
$r,r_{0}\in\Univ$.

$\d{ListMult.xs.r}{A.xs.r.(r.Zero)}$

\textbf{Theorem.}\quad{}$\mathit{ListMult.\embeds l.r}\reaches r.\embeds{\Product l}$\listmultskip

\textemdash{} Assume $r.\mathit{Zero}\conv r_{0}$.

$\d{A.xs.r.abort}{xs.(r.(S.Zero)).(B.r.abort)}$

\textbf{Lemma.}\quad{}$A.\embeds l.r.r_{0}\conv r.\embeds{\Product l}$\listmultskip

$\d{B.r.abort.x.xs}{x.abort.(C.r.abort.x.xs)}$

$\Longrightarrow\quad B.r.r_{0}.\embeds m.\embeds l\conv r.\embeds{m\cdot\Product l}$\listmultskip

$\d{C.r.abort.x.xs.x'}{A.xs.(PostMult.x.r).abort}$

$\Longrightarrow\quad C.r.r_{0}.\embeds m.\embeds l.x'\conv r.\embeds{m\cdot\Product l}$\listmultskip

$\d{PostMult.x.r.y}{Mult.x.y.r}$

$\Longrightarrow\quad\mathit{PostMult}.\embeds m.r.\embeds p\reaches r.\embeds{m\cdot p}$\listmultskip

$\d{Mult.x.y.r}{y.(r.Zero).(PostMult.x.(PostAdd.x.r))}$

\textbf{Assumption.\quad{}}$\mathit{Mult.\embeds m.\embeds p.r\reaches r.\embeds{m\cdot p}}$\listmultskip

\textbf{Usage.}\quad{}$\mathit{ListMult}.\embeds{[3,1,2]}.r\reaches r.\embeds 6$  & \textemdash{} \emph{Assume} $l,xs\in[\mathbb{N}]$, $x,x',y\in\mathbb{N}$.

$\mathit{listmult}_{4}\; l\; r=A\; l\; r\;\left(r\;0\right)$

$\Longrightarrow\quad\mathit{listmult}_{4}\; l\; r=r\;\left(\Product l\right)$\listmultskip

$\mathit{A\; l\; r\; abort}=\Caseof l$

$\quad|\;[]\rightarrow r\;1$

$\quad|\; x:xs\rightarrow\mathit{B\; r\; abort\; x\; xs}$

$\Longrightarrow\quad A\; l\; r\;\left(r\;0\right)=r\;\left(\Product l\right)$\listmultskip

$\mathit{B\; r\; abort\; x\; xs}=\Caseof x$

$\quad|\;0\rightarrow\mathit{abort}$

$\quad|\; y+1\rightarrow\mathit{C\; r\; abort\; x\; xs\; y}$

$\Longrightarrow\quad\mathit{B\; r\;\left(r\;0\right)\; x\; xs}=r\;\left(x\cdot\Product\mathit{xs}\right)$\listmultskip

$\mathit{C\; r\; abort\; x\; xs\; x'}=\mathit{A\; xs\;(PostMult\; x\; r)\; abort}$

$\Longrightarrow\quad\mathit{C\; r\;\left(r\;0\right)\; x\; xs\; x'}=r\;\left(x\cdot\Product\mathit{xs}\right)$\listmultskip

$\mathit{PostMult\; x\; r\; y}=r\;\left(x\cdot y\right)$\listmultskip

\textbf{Usage.\quad{}}$6==\mathit{listmult}_{4}\;[3,1,2]\;\mathit{id}$\tabularnewline
\end{tabular}%
}
\par\end{centering}

\caption{\label{fig:listmult}Left: `fast' list multiplication in continuation
calculus (CC). Right: Haskell program with equivalent semantics. Statements
after $\Longrightarrow$ serve to guide the reader. The theorem and
lemma are proven in Section~\ref{sec:listmult-correctness}.}
\end{figure}

\section{\label{sec:listmult-correctness}Correctness of ListMult}

\begin{multiplyingdisplayskipsby}{.4}{2.5}

This section proves that $\mathit{ListMult}$ in Figure~\ref{fig:listmult}
is correct. The idea is to assume that a program contains the listed
definitions, and $\mathit{Mult}$ behaves according to the specification;
then Theorem~\ref{thm:listmult-correctness} proves the specification
of $\mathit{ListMult}$ in that program.

We need two lemmas. Firstly, we show that name $A$ conforms to its
specification. This is done by induction on list $l$. Furthermore,
we need a lemma on the quick exit of $\mathit{PostMult}$.
\begin{lem}
The specification of $A$ is satisfied. That is, assume $l\in\mathconst{List}_{\mathbb{N}}$,
$r,r_{0}\in\Univ$ such that $r.\sembed 0\conv r_{0}$. Then $A.\embeds l.r.r_{0}\conv r.\embeds{\Product l}$. \end{lem}
\begin{proof}
We use induction on $l$, and make a three-way case distinction. 
\begin{caseenv}
\item Base case: $l=[]$. Then:
\[
\begin{array}{clcl}
 & \mathit{A.\embeds{[]}.r.r_{0}} & \quad\\
\rightarrow & \mathit{\embeds{[]}.\left(r.\left(S.Zero\right)\right).(B.r.r_{0})} &  & \mbox{by definition}\\
\reaches & \mathit{r.\left(S.Zero\right)} &  & \mbox{by definition of }\embeds{[]}\\
= & r.\embeds{\Product[]} &  & \mathit{S.Zero}\in\embeds 1=\embeds{\Product[]}
\end{array}
\]

\item $l=\left(0:l'\right)$. Then: 
\[
\begin{array}{llcl}
 & \mathit{A.\embeds{0:l'}.r.r_{0}} & \quad\\
\rightarrow & \mathit{\embeds{0:l'}.\left(r.\left(S.Zero\right)\right).(B.r.r_{0})} &  & \mbox{by definition of }A\\
\reaches & \mathit{B.r.r_{0}.\embeds 0.\embeds{l'}} &  & \mbox{by definition of }\embeds{0:l'}\\
\rightarrow & \mathit{\embeds 0.r_{0}.(C.r.r_{0}.\embeds 0.\embeds{l'})} &  & \mbox{by definition of }B\\
\reaches & r_{0} &  & \mbox{by definition of }\embeds 0\\
\conv & r.\mathit{Zero} &  & \mbox{by assumption}\\
= & r.\embeds{\Product\left(0:l'\right)} &  & \mathit{Zero}\in\embeds 0=\embeds{\Product\left(0:l'\right)}
\end{array}
\]

\item $l=\left(m+1:l'\right)$. Then:
\[
\begin{array}{llcl}
 & \mathit{A.\embeds{m+1:l'}.r.r_{0}} & \quad\\
\rightarrow & \mathit{\embeds{m+1:l'}.\left(r.\left(S.Zero\right)\right).(B.r.r_{0})} &  & \mbox{by definition of }A\\
\reaches & \mathit{B.r.r_{0}.\embeds{m+1}.\embeds{l'}} &  & \mbox{by definition of }\embeds{m+1:l'}\\
\rightarrow & \mathit{\embeds{m+1}.r_{0}.(C.r.r_{0}.\embeds{m+1}.\embeds{l'})} &  & \mbox{by definition of }B\\
\reaches & \mathit{C.r.r_{0}.\embeds{m+1}.\embeds{l'}.\embeds m} &  & \mbox{by definition of }\embeds{m+1}\\
\rightarrow & \mathit{A.\embeds{l'}.\left(PostMult.\embeds{m+1}.r\right).r_{0}} &  & \mbox{by definition of }C\\
\conv & \mathit{PostMult.\embeds{m+1}.r.\embeds{\Product l'}} &  & \mbox{by induction if }r_{0}\conv\mathit{PostMult}.\embeds{m+1}.r.\sembed 0\\
\rightarrow & \mathit{Mult.\embeds{m+1}.\embeds{\Product l'}.r} &  & \mbox{by definition of }\mathit{Postmult}\\
\reaches & r.\embeds{\left(m+1\right)\cdot\Product l'} &  & \mbox{spec }\mathit{Mult}\\
= & r.\sembed{\Product\left(m+1:l'\right)} &  & \mbox{mathematics}
\end{array}
\]

This chain proves that $A.\embeds{m+1:l'}.r.r_{0}\conv r.\embeds{\Product\left(m+1:l'\right)}$. 

\end{caseenv}

\noindent The third case requires Lemma~\ref{lem:postmult-to-rzero},
which is proved below. This completes the induction, yielding: 
\[
A.\embeds l.r.r_{0}\conv r.\embeds{\Product l}\mbox{ for all }l\in\ListNats,r\in\Univ\mbox{ .}\myqedhere
\]

\end{proof}
\begin{lem}
\label{lem:postmult-to-rzero}Let $x\in\Univ$, $l\in\mathconst{List}_{\mathbb{N}}$,
$r,r_{0}\in\Univ$ and $r.\sembed 0\conv r_{0}$. Then $\mathit{PostMult.x.r.\sembed 0\conv r_{0}}$.\end{lem}
\begin{proof}
By the following chain. 
\[
\begin{array}{clcl}
 & \mathit{PostMult}.x.r.\sembed 0 & \quad\\
\rightarrow & \mathit{Mult}.x.\sembed 0.r &  & \mbox{by definition of }\mbox{\ensuremath{\mathit{PostMult}}}\\
\rightarrow & \sembed 0.\left(r.\mathit{Zero}\right).\left(\mathit{PostMult}.x.\left(\mathit{PostAdd}.x.r\right)\right) &  & \mbox{by definition of }\mathit{Mult}\\
\rightarrow & r.\mathit{Zero}=r.\sembed 0 &  & \mbox{by definition of }\sembed 0\\
\conv & r_{0} &  & \mbox{by assumption}\myqedhere
\end{array}
\]
\end{proof}
\begin{thm}
\label{thm:listmult-correctness}The specification of $\mathit{ListMult}$
is satisfied. That is: assume $l\in\mathconst{List}_{\mathbb{N}}$,
$r\in\Univ$. Then $\mathit{ListMult}.\embeds l.r\reaches r.\embeds{\Product l}$. \end{thm}
\begin{proof}
We fill in $r_{0}=r.\mathit{Zero}$ in the specification of $A$;
then $r.\mathit{Zero}\conv r.\sembed 0$ is satisfied by definition
of $\sembed 0$. 
\[
A.\embeds l.r.\left(r.\mathit{Zero}\right)\conv r.\embeds{\Product l}\qquad\mbox{for all }l\in\ListNats,r\in\Univ
\]
 If we temporarily take $r$ to be a fresh name, then we can change
$\conv$ into $\reaches$ with Proposition~\ref{prp:converging-final-then-reaches}.
\[
A.\embeds l.r.\left(r.\mathit{Zero}\right)\reaches r.\embeds{\Product l}\qquad\mbox{for all }l\in\ListNats
\]
We can generalize this again with Theorem~\ref{thm:freshness}: 
\[
A.\embeds l.r.\left(r.\mathit{Zero}\right)\reaches r.\embeds{\Product l}\qquad\mbox{for all }l\in\ListNats,r\in\Univ
\]
 Now our main correctness result follows rather straightforwardly.
\[
\begin{array}[b]{clcl}
 & \mathit{ListMult}.\embeds l.r & \quad\\
\rightarrow & A.\embeds l.r.\left(r.\mathit{Zero}\right) &  & \mbox{by definition}\\
\reaches & r.\embeds{\Product l} &  & \mbox{as just shown}
\end{array}\myqedhere
\]

\end{proof}
\end{multiplyingdisplayskipsby}

\section{Conclusions and Future work}

We have defined a deterministic calculus that is suitable for modeling
programs with control. The calculus has simple operational semantics,
and can model both call-by-value and call-by-name programs in such
a way that CBN and CBV subprograms can be combined. Call-by-push-value~\cite{Levy99}
is another calculus in which CBV and CBN lambda calculus can be embedded.
We believe that CBV and CBN subprograms can be meaningfully combined
in call-by-push-value, but we have not found this in the literature.

In the present paper, we have not yet exploited types. In the future,
we will develop a typed version of continuation calculus, which also
guarantees the termination of well-typed terms. Another way to look
at types is by giving a standard representation of data terms as terms
in continuation calculus. In this paper, we have shown how to do this
for booleans, natural numbers and lists; in future work we will extend
this to other (algebraic, higher order, \ldots{}) data types. Also,
we will develop a generic procedure to transform functions that are
defined by pattern matching and equations into terms of continuation
calculus.

The determinism in continuation calculus suggests that we can model
assignment and side effects using a small number of extra names with
special reduction rules. However, such an extension may not preserve
observational equivalence. We want to examine if an extension provides
a pragmatic model for imperative-functional garbage-collected languages,
such as OCaml.

\bibliographystyle{eptcs}
\bibliography{bib}
\clearpage{}

\appendix

\section{Proofs}

We first prove the theorems in Section~\ref{sub:theorems-general},
then those in Section~\ref{sub:theorems-program-substitution-union},
and finally those in Section~\ref{sub:theorems-equivalence}. The
theorems within a subsection are not proved in order, and are interspersed
with lemmas.

\subsection{General}
\begin{prop}
\label{prp:fresh-reducts}Let name $\place$ be not mentioned in term
$M$ and program $P$. Then $\place$ is not mentioned in any reduct
of $M$.\end{prop}
\begin{proof}
By induction and by definition of $\next{}M$.
\end{proof}

\begin{thm}
Let $M,N\in\Univ$, $P$ a program, and $\place$ a name not mentioned
in $P$. The following equivalences hold: 
\[
\begin{array}{rclcr}
M\rightarrow N & \Longleftrightarrow & \forall t\in\Univ:M[\place:=t]\rightarrow N[\place:=t] & \qquad & (1)\\
M\fin & \Longleftrightarrow & \exists t\in\Univ:M[\place:=t]\fin &  & (2)\\
M\reaches N & \Longleftrightarrow & \forall t\in\Univ:M[\place:=t]\reaches N[\place:=t] &  & (3)\\
M\evfinal & \Longleftrightarrow & \exists t\in\Univ:M[\place:=t]\evfinal &  & (4)
\end{array}
\]

\end{thm}
This theorem implies Theorem~\ref{thm:freshness}.
\begin{proof}
~
\begin{elabeling}{00.00.0000}
\item [{($\Leftarrow$1).}] Fill in $t=\place$.
\item [{($\Rightarrow$1).}] Since $\next PM$ exists, $\head M$ must
be in the domain of $P$. Because $\place\notin\dom P$, we know $\head M\ne\place$.
Let $M=n.u_{1}.\cdots.u_{k}$ and $\mbox{``}n.x_{1}.\cdots.x_{k}\rightarrow r\mbox{''}\in P$,
where $n$ is a name. Then $M[\place:=t]=n.u_{1}[\place:=t].\cdots.u_{k}[\place:=t]\rightarrow r[\vec{x}:=\vec{u}\left[\place:=t\right]]$.
Since $\place$ is not mentioned in $r$, the last term is equal to
$r[\vec{x}:=\vec{u}][\place:=t]=N[\place:=t]$.
\item [{($\Leftarrow$2).}] Assume $M[\place:=t]\fin$. Then $\head{M[\place:=t]}\notin\dom P$
or $\length{M[\place:=t]}\ne\arity{}{\head{M[\place:=t]}}$. If $\head M=\place$,
then $M\fin$, so assume $\head M\ne\place$. Then $\head M=\head{M[\place:=t]}$
and $\length M=\length{M[\place:=t]}$, so also $M\fin$.
\item [{($\Rightarrow$2,~$\Leftarrow$3,~$\Rightarrow$4).}] Fill in
$t=\place$.
\item [{($\Rightarrow$3).}] $\reaches$ is the reflexive and transitive
closure of $\rightarrow$.
\item [{($\Leftarrow$4).}] Suppose that $M[\place:=t]\reaches N\fin$
in $k$ steps. If on the contrary $M$ is not terminating, then $M\reaches M'$
in $k+1$ steps. By repeated application of ($\Rightarrow$1), also
$M[\place:=t]\reaches M'[\place:=t]$ in $k+1$ steps. Contradiction.
\qedhere 
\end{elabeling}
\end{proof}

\begin{proof}[Proof of Lemma~\ref{lem:determinism} (determinism)]

We assumed that $M\reaches m.\vec{t}\fin$ and $M\reaches n.\vec{u}\fin$.
So $m.\vec{t}$ and $n.\vec{u}$ are the term at the end of the execution
path of $M$; we see that they must be equal.
\end{proof}

\begin{proof}[Proof of Proposition~\ref{prp:converging-final-then-reaches} ($M\reaches t\backreaches N\fin$
then $M\reaches N$)]
By assumption, $M\reaches t\backreaches N$. If $N\reaches t$ in
1 or more steps, then we could not have had $N\fin$. Thus $N\reaches t$
in 0 steps: $N=t$.
\end{proof}

\subsection{Program substitution and union}

These are proofs of theorems in Section~\ref{sub:theorems-program-substitution-union}.
\begin{proof}[Proof of Theorem~\ref{thm:program-transformation}, equivalences
1/3]
Let $M=h.t_{1}.\cdots.t_{k}$, where $h$ is a name. We have the
following cases.
\begin{enumerate}
\item \emph{$h\notin\dom P$}. All names in $\dom{P\sigma}\setminus\dom P$
are fresh, so also $h\sigma\notin\dom{P\sigma}$. We see that both
$\next PM$ and $\next{P\sigma}{M\sigma}$ are undefined.
\item \emph{$h\in\dom P$}. Then some rule ``$\d{h.x_{1}.\cdots.x_{l}}r$''
is in $P$, while ``$\d{h\sigma.x_{1}.\cdots.x_{l}}{r\sigma}$''
is in $P\sigma$. If $k\ne l$, then both $\next PM$ and $\next{P\sigma}{M\sigma}$
are undefined.

If $k=l$, then $\next PM=r[\vec{x}:=\vec{t}]$, and $\next{P\sigma}{M\sigma}=r\sigma[\vec{x}:=\vec{t}\sigma]$.
We note that the domains of $\sigma$ and $[\vec{x}:=\vec{t}\sigma]$
are disjoint because names are never variables. We can therefore do
the substitutions in parallel: $r\sigma[\vec{x}:=\vec{t}\sigma]=r\left[\vec{n}:=\vec{m},\vec{x}:=\vec{t}\sigma\right]$.
Because the result of $\sigma$ is never in $\dom{\sigma}$ (all $m_{i}$
are fresh), we can even put $\left[\vec{n}:=\vec{m}\right]$ at the
end: $r\left[\vec{n}:=\vec{m},\vec{x}:=\vec{t}\sigma\right]=r\left[\vec{x}:=\vec{t}\sigma\right]\left[\vec{n}:=\vec{m}\right]=r\left[\vec{x}:=\vec{t}\sigma\right]\sigma$.
Also, because the result of $\sigma$ is never in $\dom{\sigma}$
, we know that $\sigma\sigma=\sigma$. We find that $r\left[\vec{x}:=\vec{t}\sigma\right]\sigma=r\left[\vec{x}:=\vec{t}\right]\sigma=N\sigma$.
This completes the proof. \qedhere 

\end{enumerate}
\end{proof}

\begin{proof}[Proof of Theorem~\ref{thm:program-transformation}, equivalences
2/4]
The second equivalence is by transitivity of equivalence 1. The fourth
equivalence is then trivial.
\end{proof}

\begin{proof}[Proof of Theorem~\ref{thm:program-transformation}, implications
1--4]
$\next PM$ exists iff a rule $\in P$ defines it; by $P\subseteq P'$
that rule also defines $\next{P'}M$. This proves implication 1 and
3. The second implication follows using the structure of $M\twoheadrightarrow_{P}N$.
Then the fourth implication follows trivially.
\end{proof}

\begin{proof}[Proof of Theorem~\ref{thm:program-transformation}, implication 5]
We have to show that for all $P''\supseteq P'$ and $X\in\Univ$,
$X.M\evfinal_{P''}\Longleftrightarrow X.N\evfinal_{P''}$. This follows
from $M\obseq_{P}N$ because $P\subseteq P'\subseteq P''$.
\end{proof}

\begin{proof}[Proof of Theorem~\ref{thm:program-transformation}, equivalence 5]
We show the left-implication. The right implication then follows
from $M\sigma\sigma^{-1}\obseq_{P\sigma\sigma^{-1}}N\sigma\sigma^{-1}\Leftarrow M\sigma\obseq_{P\sigma}N\sigma$,
because $\sigma\sigma^{-1}$ is the identity substitution.

So suppose $M\sigma\obseq_{P\sigma}N\sigma$, and let program $Q\supseteq P$
and term $X$ be given. We prove $X.M\evfinal_{Q}\Leftrightarrow X.N\evfinal_{Q}$
by the following chain: 
\[
\begin{array}[b]{clcl}
 & X.M\evfinal_{Q} & \qquad\\
\Leftrightarrow & X\sigma.M\sigma\evfinal_{Q\sigma} &  & \mbox{(Theorem \ref{thm:program-transformation} equivalence 2/3)}\\
\Leftrightarrow & X\sigma.N\sigma\evfinal_{Q\sigma} &  & \mbox{(\ensuremath{M\obseq_{P\sigma}N}and \ensuremath{Q\sigma\supseteq P\sigma})}\\
\Leftrightarrow & X.N\evfinal_{Q} &  & \mbox{(Theorem \ref{thm:program-transformation} equivalence 2/3)}
\end{array}\myqedhere
\]

\end{proof}

\begin{lem}
\label{lem:evfinal-program-shortening}Assume program $P'\supseteq P$
and $M\in\Univ$ such that $\dom{P'\setminus P}$ is not mentioned
in $P$ or $M$. Then $M\evfinal_{P}\Leftrightarrow M\evfinal_{P'}$.\end{lem}
\begin{proof}
Regard $\next PM$ and $\next{P'}M$. The names in $\dom{P'\setminus P}$
are not mentioned in $M$, so either both $\next PM$ and $\next{P'}M$
are defined and equal, or both are undefined. The names in $\dom{P'\setminus P}$
are still not mentioned in $M$'s successor, so the previous sentence
applies to all reducts of $M$. We find that $M$ has a final reduct
in $P$ iff it has one in $P'$, hence $M\evfinal_{P}\Leftrightarrow M\evfinal_{P'}$.
\end{proof}

\begin{proof}[Proof of Theorem~\ref{thm:equiv-program-shortening-lengthening}]
The right-implication is already proven by Theorem~\ref{thm:program-transformation},
so we prove the left-implication.

Suppose program $P'\supseteq P$, but $\dom{P'\setminus P}$ is not
mentioned in $M,N$. Suppose furthermore program $Q\supseteq P$ and
$X\in\Univ$. Then we have to prove $X.M\evfinal_{Q}\Leftrightarrow X.N\evfinal_{Q}$.
$Q$ is not required to be a superset of $P'$; it may even define
some names differently than $P'$.

Although we know that $\Delta=\dom{P'\setminus P}$ is not used in
$M$ or $N$, any name $\in\Delta$ could be used in $X$. We want
to compare $X.M$ and $X.N$ on an extension program of $Q$, so we
will make sure that $X$ does not accidentally refer to names in $\Delta$.
We will rename all $d\in\Delta$ within $X$ and $P'$.

Take a substitution $\sigma=[d_{i}:=d'_{i}|d_{i}\in\Delta]$ that
renames all $d\in\Delta$ to fresh names for $M,N,X,P',Q$. We know
that $M=M\sigma$, $N=N\sigma$, and $P=P\sigma$, because all $d\in\Delta$
are not mentioned in $M$, $N$, or $P$. Now note that $\left(X.M\right)\sigma=X\sigma.M$
and $\left(X.N\right)\sigma=X\sigma.N$ do not contain a name in $\Delta$,
nor does any such name occur in $Q\sigma$. 

Take $Q'=P'\cup Q\sigma$. Then $Q'$ is a program because $\dom{Q\sigma\setminus P}$
has no overlap with $\dom{P'\setminus P}=\Delta$. Furthermore, $Q'$
is an extension program of both $P'$ and $Q\sigma$. We apply Lemma~\ref{lem:evfinal-program-shortening}
to see that 
\begin{equation}
\begin{array}{cc}
 & X\sigma.M\evfinal_{Q\sigma}\Leftrightarrow X\sigma.M\evfinal_{Q'}\\
\mbox{and} & X\sigma.N\evfinal_{Q\sigma}\Leftrightarrow X\sigma.N\evfinal_{Q'}\,.
\end{array}\label{eq:equivalence-program-shortening-1}
\end{equation}
 We can thus make the following series of bi-implications.

\[
\begin{array}{rcll}
X.M\evfinal_{Q} & \Leftrightarrow & X\sigma.M\evfinal_{Q\sigma} & \qquad\mbox{(Theorem \ref{thm:program-transformation})}\\
 & \Leftrightarrow & X\sigma.M\evfinal_{Q'} & \qquad\eqref{eq:equivalence-program-shortening-1}\\
 & \Leftrightarrow & X\sigma.N\evfinal_{Q'} & \qquad\mbox{(}M\obseq_{P'}N,P'\subseteq Q'\mbox{)}\\
 & \Leftrightarrow & X\sigma.N\evfinal_{Q\sigma} & \qquad\eqref{eq:equivalence-program-shortening-1}\\
 & \Leftrightarrow & X.N\evfinal_{Q} & \qquad\mbox{(Theorem \ref{thm:program-transformation})}
\end{array}
\]

Because we showed $X.M\evfinal_{Q}\Leftrightarrow X.N\evfinal_{Q}$,
we can conclude that $M\obseq_{P}N$.
\end{proof}

\subsection{Term equivalence}
\begin{prop}
$\conv$ is an equivalence relation.\end{prop}
\begin{proof}
Suppose $A\conv C$ and $C\conv E$, then there exist $B,D$ such
that $A\reaches B\backreaches C\reaches D\backreaches E$. Suppose
that $C\reaches B$ in $k$ steps and $C\reaches D$ in $l$ steps.
Without loss of generality, $k\le l$. By determinism of $\rightarrow$,
\[
C\underset{k\text{ steps}}{\reaches}B\underset{l-k\text{ steps}}{\reaches}D.
\]
 Then $A\reaches B\reaches D$. We see that $D$ is a common reduct
of $A$ and $E$.
\end{proof}

\begin{prop}
\label{prp:equiv-equivalence}$\obseq$ is an equivalence relation.\end{prop}
\begin{proof}
Reflexivity and symmetry are trivial. We have to prove transitivity:
if $M\obseq_{P}N$ and $N\obseq_{P}O$, and $P\subseteq P'$, then
$X.M\evfinal_{P'}\Leftrightarrow X.O\evfinal_{P'}$. We know from
the premises that $X.M\evfinal_{P'}\Leftrightarrow X.N\evfinal_{P'}$
and $X.N\evfinal_{P'}\Leftrightarrow X.O\evfinal_{P'}$. 
\end{proof}

\begin{lem}
\label{lem:next-then-evfinal}If $X\reaches Y$, then $X\evfinal\Leftrightarrow Y\evfinal$.\end{lem}
\begin{proof}
By induction on the number of steps $s$ in $X\reaches Y$. If $X=Y$,
then trivial, so assume $s\ge1$. This implies the existence of term
$X'$ such that $X\rightarrow X'$.

If there exists $Z$ such that $Y\reaches Z\fin$, then $X\reaches Y\reaches Z\fin$.
Reversely, assume $X\reaches Z\fin$ for some $Z$. Because $X\ne Y$
and by determinism of $\rightarrow$ we know $X\rightarrow X'\reaches Z\fin$
and $X\rightarrow X'\reaches Y$. By induction on $X'\reaches Y$
we get $Y\reaches Z\fin$.
\end{proof}

\begin{proof}[Proof of Proposition~\ref{prp:equiv-then-same-termination} ($M\obseq N$
then $M\evfinal\Leftrightarrow N\evfinal$)]
Take a fresh name $X$, and define $P'=P\cup\left\{ \d{X.t}t\right\} $.
\[
\begin{array}{ccll}
M\evfinal_{P} & \Leftrightarrow & M\evfinal_{P'} & \qquad\mbox{(evaluation of }M\mbox{ never contains a head in \ensuremath{\dom{P'\setminus P}})}\\
 & \Leftrightarrow & X.M\evfinal_{P'} & \qquad\mbox{(}X.M\rightarrow_{P'}M\mbox{, }\rightarrow\mbox{ deterministic)}\\
 & \Leftrightarrow & X.N\evfinal_{P'} & \qquad\mbox{(}M\obseq_{P}N\mbox{)}\\
 & \Leftrightarrow & N\evfinal_{P'} & \qquad\mbox{(}X.N\rightarrow_{P'}N\mbox{, }\rightarrow\mbox{ deterministic)}\\
 & \Leftrightarrow & N\evfinal_{P} & \qquad\mbox{(evaluation of }N\mbox{ never contains a head in \ensuremath{\dom{P'\setminus P}})}
\end{array}
\]

\end{proof}

\begin{lem}
\label{lem:dot-of-reducible-not-reducible}If $X\rightarrow Y$, then
$X.\vec{t}\fin$ for $k>0$.\end{lem}
\begin{proof}
$\next{}M$ exists iff the length of the corresponding left-hand side
is equal to $\length M$, and $\length{X.t_{1}.\cdots.t_{k}}=\length X+k$.
The corresponding left-hand side is the same for $X$ and $X.\vec{t}$.
\end{proof}

\begin{proof}[Proof of Lemma~\ref{lem:equiv-congruence} ($\obseq$ is a congruence)]
Let $P'\supseteq P$ be an extension program. We must prove that
for all $X$, $X.(M.N)\evfinal_{P'}\Leftrightarrow X.(M'.N')\evfinal_{P'}$.
Extend $P'$ to $P''=P'\cup\{\d{A.m.n}{X.(m.n)},\d{B.n.m}{X.(m.n)}\}$.
Note that by Lemma~\ref{lem:next-then-evfinal}, 
\[
\begin{array}{cccccc}
 & X.\left(M.N\right)\evfinal_{P''} & \Leftrightarrow & A.M.N\evfinal_{P''} & \Leftrightarrow & B.N.M\evfinal_{P''}\\
\mbox{and } & X.\left(M'.N'\right)\evfinal_{P''} & \Leftrightarrow & A.M'.N'\evfinal_{P''} & \Leftrightarrow & B.N'.M'\evfinal_{P''}\mbox{,}
\end{array}
\]
so we can make the following chain: 
\[
\begin{array}{clcl}
 & X.\left(M.N\right)\evfinal_{P''} & \quad\\
\Leftrightarrow & A.M.N\evfinal_{P''}\\
\Leftrightarrow & A.M.N'\evfinal_{P''} &  & \mbox{(}N\obseq N'\mbox{)}\\
\Leftrightarrow & B.N'.M\evfinal_{P''}\\
\Leftrightarrow & B.N'.M'\evfinal_{P''} &  & \mbox{(}M\obseq M'\mbox{)}\\
\Leftrightarrow & X.\left(M'.N'\right)\evfinal_{P''}\,.
\end{array}
\]
Now by Lemma~\ref{lem:evfinal-program-shortening}, $X.\left(M.N\right)\evfinal_{P'}\Leftrightarrow X.\left(M'.N'\right)\evfinal_{P'}$,
which was to be shown.
\end{proof}

\begin{lem}
\label{lem:common-ext-reduct-then-evfinal}Let $M,N\in\Univ$ and
$k\ge0$. Let names $\place_{1},\cdots,\place_{k}$ be not mentioned
in $M,N,P$. Suppose $M.\place_{1}.\cdots.\place_{k}\rightarrow M'\reaches t\backreaches N'\leftarrow N.\place_{1}.\cdots.\place_{k}$.
Let name $\otherplace$ be not mentioned in $M,N,P$. Then $\forall X\in\Univ:X[\otherplace:=M]\evfinal\Leftrightarrow X[\otherplace:=N]\evfinal$.\end{lem}
\begin{proof}
Suppose that $X[\otherplace:=M]\reaches X'\fin$ in $n$ steps. We
will show that $X[\otherplace:=N]\evfinal$. The other direction holds
by symmetry. The proof goes by induction on $n$.

Because $\next{}{M.\place_{1}.\cdots.\place_{k}}$ and $\next{}{N.\place_{1}.\cdots.\place_{k}}$
exist, we know that $\arity{}M=\arity{}N=k$. Regard $\head X$. We
distinguish four cases:
\begin{caseenv}
\item \emph{$\head X=\otherplace$, and $\length X\ne k$}. Then $\arity{}{X[\otherplace:=N]}$
is undefined or not zero, hence $X[\otherplace:=N]\fin$.
\item \emph{$\head X=\otherplace$, and $\length X=k$}. Then there exist
$u_{1},\cdots,u_{k}$ such that $X=\otherplace.u_{1}.\cdots.u_{k}$.
Then: 
\[
\begin{aligned}X[\otherplace:=M] & =M.u_{1}[\otherplace:=M].\cdots.u_{l}[\otherplace:=M] & \quad\\
 & \rightarrow M'[\vecplace:=\vec{u}[\otherplace:=M]] &  & \mbox{(}\vecplace\mbox{ fresh for }M,P\mbox{)}\\
 & =M'[\vecplace:=\vec{u}][\otherplace:=M] &  & \mbox{(}\mbox{\ensuremath{\otherplace}\ fresh for }M'\mbox{)}\\
 & \reaches t[\vecplace:=\vec{u}][\otherplace:=M] &  & \mbox{(}\vecplace\mbox{ fresh for }M,P\mbox{)}
\end{aligned}
\]
Analogously, $X[\otherplace:=N]\reaches t[\vecplace:=\vec{u}][\otherplace:=N]$.
We know $t[\vecplace:=\vec{u}][\otherplace:=M]\reaches X'\fin$ in
at most $n-1$ steps, and $\otherplace$ is not mentioned in $t[\vecplace:=\vec{u}]$,
so using the induction hypothesis we get $t[\vecplace:=\vec{u}][\otherplace:=N]\evfinal$.
Hence, $X[\otherplace:=N]\evfinal$.
\item \emph{$\head X\ne\otherplace$, and }$\arity{}X\ne0$\emph{ or undefined}.
Then $\arity{}{X[\otherplace:=M]}=\arity{}{X[\otherplace:=N]}=\arity{}X\ne0$
or undefined, hence $X[\otherplace:=M]\fin$ and $X[\otherplace:=N]\fin$.
\item \emph{$\head X\ne\otherplace$, and $\arity{}X=0$}. Then $\next{}{X[\otherplace:=M]}=\next{}X[\otherplace:=M]$
and $\next{}{X[\otherplace:=N]}=\next{}X[\otherplace:=N]$. We assumed
$X[\otherplace:=M]\reaches X'\fin$ in $n$ steps, so $\next{}X[\otherplace:=M]\reaches X'\fin$
in at most $n-1$ steps. We can therefore apply the induction hypothesis
to find $\next{}X[\otherplace:=N]\evfinal$. \qedhere 
\end{caseenv}
\end{proof}

\begin{proof}[Proof of Theorem~\ref{thm:converging-ext-then-equiv}]
Suppose $P'\supseteq P$ is an extension program. We must prove $X.M\evfinal\Leftrightarrow X.N\evfinal$.

Take $\vec{\place}$ fresh for $P',X,M,N$. Because $\arity{}{M.\place_{1}.\cdots.\place_{k}}=\arity{}{N.\place_{1}.\cdots.\place_{k}}=0$,
they both have a successor, say $M'$ and $N'$. By definition of
$\conv$ and determinism of $\rightarrow$, we know $M.\place_{1}.\cdots.\place_{k}\rightarrow M'\reaches t\backreaches N'\leftarrow N.\place_{1}.\cdots.\place_{k}$.
Then by Lemma~\ref{lem:common-ext-reduct-then-evfinal}, we know
$X.M\evfinal_{P'}\Leftrightarrow X.N\evfinal_{P'}$.
\end{proof}

\begin{proof}[Proof of Theorem~\ref{thm:obseq-then-reaches-same-name} ($M\reaches\place.t_{1}.\cdots.t_{k}$,
$M\obseq N$ then $N\reaches\place.u_{1}.\cdots.u_{k}$)]

Because $\place\notin\dom P$, we know $M\reaches\place\fin$. By
Proposition~\ref{prp:equiv-then-same-termination}, $N\reaches N'\fin$.

Suppose on the contrary that $\head{N'}\ne\place$ or $\length{N'}\ne k$.
We will deduce an impossibility. Make an extension program $P'=P\cup\{\d{\place.x_{1}.\cdots.x_{k}}{\place.x_{1}.\cdots.x_{k}}\}$.
Then $M\reaches\place.t_{1}.\cdots.t_{k}$ is nonterminating under
$P'$. But by definition of $\plainnext$, $N'\fin_{P'}$. This contradicts
with Proposition~\ref{prp:equiv-then-same-termination} and Theorem~\ref{thm:equiv-program-shortening-lengthening},
which prove that $N\reaches N'$ is nonterminating. Hence we conclude
$N'=\place.u_{1}.\cdots.u_{k}$ for some terms $\vec{u}$.
\end{proof}

\begin{proof}[Proof of Proposition~\ref{prp:all-nats-eqv-succs-zero}]
Let $\denot{\mathbb{N}_{0}}=\{M\in\Univ|\forall z,s\in\Univ:M.z.s\reaches z\}$
and $\denot{\mathbb{N}_{k+1}}=\{M\in\Univ|\exists x\in\denot{\mathbb{N}_{k}}\;\forall z,s\in\Univ:M.z.s\reaches s.x\}$.
Then every $\denot{\mathbb{N}_{k}}\subseteq\denot{\mathbb{N}}$, and
$\cup_{k\in\mathbb{N}}\denot{\mathbb{N}_{k}}$ satisfies the defining
equation of $\denot{\mathbb{N}}$, so $\denot{\mathbb{N}}=\cup_{k\in\mathbb{N}}\denot{\mathbb{N}_{k}}$.

Suppose $M\in\denot{\mathbb{N}}$, then $M\in\mbox{some }\denot{\mathbb{N}_{k}}$.
We proceed by induction on $k$. If $k=0$, then Theorem~\ref{thm:converging-ext-then-equiv}
shows $M\obseq\sembed 0$. If $k\ge1$, then by the induction hypothesis
there is some $x\in\denot{\mathbb{N}_{k-1}}$ such that for all $z,s$,
$M.z.s\reaches s.x$. Observe that $M.z.s\conv S.x.z.s$ for all $z,s$,
so Theorem~\ref{thm:converging-ext-then-equiv} shows us $M\obseq S.x$.
We get $S.x\obseq S.\sembed{k-1}=\sembed k$ by the induction hypothesis
and Lemma~\ref{lem:equiv-congruence}, from which we get the result.\end{proof}

\end{document}